\documentclass[reqno,11pt,a4paper]{amsart}

\usepackage[numbers]{natbib}

\usepackage[linktocpage=true,colorlinks=true, linkcolor=blue, citecolor=red, urlcolor=magenta]{hyperref}

\usepackage{amsmath, amsthm, amsfonts, amssymb,bbm}
\usepackage{mathtools,empheq}
\usepackage[all]{xy}
\usepackage{cancel}
\usepackage{enumerate}
\usepackage{comment}

\numberwithin{equation}{section}
\usepackage[utf8]{inputenc}

\theoremstyle{plain}
\newtheorem{theorem}{Theorem}
\newtheorem{proposition}[theorem]{Proposition}
\newtheorem{lemma}[theorem]{Lemma}

\theoremstyle{definition}

\newtheorem{definition}[theorem]{Definition}

\theoremstyle{remark}
\newtheorem{remark}[theorem]{Remark}

\DeclareMathOperator{\Mat}{Mat}
\DeclareMathOperator{\rmat}{r}
\DeclareMathOperator{\kmat}{k}
\DeclareMathOperator{\aff}{Aff}
\DeclareMathOperator{\Vect}{Vect}

\def\Z{\mathbb{Z}}	
\def\C{\mathbb{C}}	
\def\R{\mathbb{R}}	
\def\cS{\mathcal{S}}
\renewcommand{\leq}{\leqslant} 		
\renewcommand{\geq}{\geqslant}
\newcommand{\mc}[1]{\mathcal{#1}}
\newcommand{\mf}[1]{\mathfrak{#1}}
\newcommand{\mb}[1]{\mathbb{#1}}
\newcommand{\id}{\mathbbm{1}}

\def\ud{\mathrm{d}}

\newcommand{\dev}{\partial}

\let\ker\relax
\DeclareMathOperator{\ker}{Ker}
\def\A{\mathcal{A}}
\def\hA{\hat{\A}}
\def\F{\mathcal{F}}
\def\hF{\hat{\F}}

\usepackage{xcolor}
\definecolor{light}{gray}{.9}

\begin{document}

\title{Multi--component Hamiltonian difference operators}
\author{Matteo Casati}
\address{School of Mathematics and Statistics\\Ningbo University\\Ningbo City, Zhejiang Province, PRC}
\email{matteo@nbu.edu.cn}
\author{Daniele Valeri}
\address{Dipartimento di Matematica \& INFN, Sapienza Universit\`a di Roma,
P.le Aldo Moro 5, 00185 Rome, Italy}
\email{daniele.valeri@uniroma1.it}

\begin{abstract}
In this paper we study local Hamiltonian operators for multi-component evolutionary differential-difference equations. We address two main problems: the first one is the classification of low order operators for the 
two-component case. On the one hand, this extends the previously known results 
in the scalar case; on the other hand, our results include the degenerate cases, going beyond the foundational investigation conducted by Dubrovin. The second problem is the study and the computation of the Poisson cohomology for a two-component $(-1,1)$-order Hamiltonian operator with degenerate leading term appearing in many integrable differential-difference systems, notably the Toda lattice. The study of its Poisson cohomology sheds light on its deformation theory and the structure of the bi-Hamiltonian pairs where it is included in, as we demonstrate in a series of examples.
\end{abstract}

\maketitle

\tableofcontents

\section{Introduction}
%
%
Evolutionary differential-difference equations (D$\Delta$Es) are a class of systems for functions depending on two sets of variables, of which one is a continuous parameter (the \emph{time}) and the other takes value on a lattice. The prototypical example of such class is the (infinite) Volterra lattice, satisfied by a function $u(n,t)$ ($n \in \Z$, $t\in\R$) such that
$$
\dev_t u(n,t)=u(n,t)\left(u(n+1,t)-u(n-1,t)\right).
$$
Introducing the shift operator $\cS$ such that $\cS^p u(n,t)=u(n+p,t)$, $p\in\mb Z$, a generic D$\Delta$E system for $\ell$ functions (we say that the system has $\ell$ components) of one spatial and one time variable is of the form ($u=(u^1,\dots,u^\ell)$)
\begin{equation*}
\frac{\dev}{\dev t}u^i(n,t)=F^i\left(\ldots, \cS^{-1}u, u,\cS u,\ldots\right)
\end{equation*}
for $i=1,\ldots,\ell$. Among this class of systems, we are particularly interested to those who can be cast in
a local \emph{Hamiltonian form}, namely written as
\begin{equation}\label{intro:Ham}
\dev_t u^i=\sum_{j=1}^\ell K^{ij}\frac{\delta H}{\delta u^j}
\end{equation}
for a difference operator $K$ and a local functional $H$. The operator $K$ must be such that it defines a Lie algebra structure (the \emph{local} Poisson bracket) on the space of local functionals and the usual Lie algebra morphism between local functionals and the Lie algebra of evolutionary vector fields; if this is the case, we say that $K$ is a \emph{Hamiltonian operator}.

The study of Hamiltonian operators is particularly important in the theory of integrable
systems and in deformation quantisation. The Hamiltonian structures of many integrable differential-difference systems have been identified (see \cite{kmw13} and references therein); as it is well known, the existence of two different, but ``compatible'', Hamiltonian formulations for the same system  makes it bi-Hamiltonian -- this is a criterion to assess its integrability and a tool to obtain its symmetries \cite{magri, ol87}.

The study of Hamiltonian structures of D$\Delta$Es is the main topic of this paper. The theory of multiplicative Poisson vertex algebras \cite{DSKVW18-1,DSKVW18-2} has provided a convenient framework for their description. In \cite{DSKVW18-1}, in particular, the second author and collaborators obtained a classification of $\ell=1$ components (or \emph{scalar}) difference Hamiltonian operators up to the 5th order. 

Little inspection, however, has been devoted to multi-component difference Hamiltonian structures. Despite the relative abundance of examples, the only classification result we are aware of is an early result by Dubrovin \cite{Dub89, DN89}, and later expanded on by Parodi \cite{Parodi}. These authors addressed the case of  difference Hamiltonian operators of order $(-1,1)$ with nondegenerate leading term, namely difference operators as in \eqref{eq:diffop-1st} where $A$ is a nondegenerate matrix.

In this paper, we provide necessary and sufficient conditions for the existence of multi-component Hamiltonian structures of order $(-1,1)$ depending at most on nearest neighbours (see Theorem \ref{thm:1st-ham}), which itself is a necessary condition in the nondegenerate case (see Proposition \ref{prop:dep}).  For the two-component case ($\ell=2$) we provide normal forms for the Hamiltonian operators, in both the nondegenerate and degenerate cases. The prototypical example in this class is given by the Toda lattice (see \cite{kmw13} and Section \ref{sec:toda} hereinafter for more details)
\begin{equation*}
\left\{\begin{array}{rl}
u_t&=u\left(v_1-v\right)\\
v_t&=u-u_{-1}
\end{array}\right.
\end{equation*} 
which can be written in Hamiltonian form \eqref{intro:Ham} with respect to the two Hamiltonian operators
\begin{align*}
H_1&=\begin{pmatrix}0 & u\cS-u\\u-u_{-1}\cS^{-1}&0\end{pmatrix}, & H_2&=\begin{pmatrix}uu_1\cS-u_{-1}u\cS^{-1}&uv_1\cS-uv\\uv-u_{-1}v\cS^{-1}&u\cS-u_{-1}\cS^{-1}\end{pmatrix}.
\end{align*} 
The second Hamiltonian structure $H_2$ belongs to the class first studied by Dubrovin \cite{Dub89}, while it is worthy noticing that the first Hamiltonian structure has degenerate leading term and, in some coordinate system, assumes the constant form
\begin{equation}\label{eq:intro1}
H_0=\begin{pmatrix}0&\cS-1\\1-\cS^{-1}&0\end{pmatrix}.
\end{equation}

In analogy with finite dimensional manifolds, where Poisson brackets are identified with Poisson bivectors, we can do the same on the infinite-dimensional manifolds modelling PDEs and D$\Delta$Es.
The Poisson bivectors associated to the Hamiltonian operators can be used to define a differential on the complex of multivectors, whose cohomology is called the Poisson (or Poisson-Lichnerowicz) cohomology \cite{l77}. It provides information about the center of the Poisson algebra (the Casimir
functions), its symmetries, and the compatible bivectors that can be defined on the same
manifold.

The computation of the full Poisson cohomology for a given Poisson bivector is, in general, a formidable task. In his seminal work on Poisson manifolds \cite{l77}, Lichnerowicz established an isomorphism between the Poisson cohomology of a \emph{symplectic} manifold and the De Rham cohomology of the same manifold. In the differential case, the main result on the cohomology of Hamiltonian operators is due to Getzler, \cite{G01} who computed the Poisson cohomology of operators of hydrodynamic type, i.e. homogeneous of first order. The first author has contributed to investigate higher dimensional versions of the same problem \cite{ccs1,ccs2}. The notion of Poisson cohomology for difference Hamiltonian operators has been introduced in \cite{CW19}, where it was computed in the case of scalar ($\ell=1$) operators of order $(-1,1)$, showing that it is essentially trivial and hence reproducing in a different setting Getzler's result. Indeed, operators of order $(-1,1)$ can be seen as discretizations of differential operators of hydrodynamic type.

In this paper, we compute the Poisson cohomology for difference Hamiltonian structures of order $(0,0)$ (with $\ell=2n$) and for the first Hamiltonian structure of the Toda lattice (for which we have the constant normal form $H_0$ in \eqref{eq:intro1}). In particular, we show that its second cohomology group is essentially trivial, namely all the compatible structures of order $(-N,N)$, for $N\geq 1$ can be obtained from $H_0$ by a Miura transformation.

The paper is organized as follows: in Section \ref{sec2} we recall the formalism we use to describe D$\Delta$Es and their Hamiltonian structures, both according to the theory of multiplicative Poisson vertex algebras \cite{DSKVW18-1,DSKVW18-2} and the so-called $\theta$-formalism for (local) poly-vectors. In Section \ref{sec:ham} we find necessary and sufficient conditions for a (matrix) difference operator of order $(-1,1)$ to be Hamiltonian, and in particular we classify two-component operators depending only on nearest neighbours, including the class originally studied by Dubrovin, Novikov, and Parodi. In Section \ref{sec:coho}, after recalling the notion of Poisson cohomology for difference Hamiltonian operators, we compute it for the so-called \emph{ultralocal case} (namely a difference operator of order $(0,0)$) and for the first Hamiltonian structure of the Today lattice. We find (see Theorem \ref{thm:cohoToda}) that the Poisson cohomology $H^p(\hat{\F},d_{P_0})$ is trivial for $p>2$, while it is concentrated in the ``ultralocal'' sector for lower $p$. In particular, as previously mentioned, this implies that higher order deformations of $H_0$ are trivial, i.e. related to $H_0$ by a Miura transformation. Finally, in Section \ref{sec:eg} we illustrate the occurrence of the aforementioned triviality for several examples taken from the review paper \cite{kmw13}.

\section{Multiplicative Poisson vertex algebras and poly-vectors}\label{sec2}
In this section we recall the definition of a multiplicative Poisson vertex algebra \cite{DSKVW18-1} and of the so-called $\theta$-formalism for difference poly-vectors \cite{CW19}.

The notion of a Hamiltonian structure described in these two frameworks is the same \cite{CW19} and it is equivalent with the original definition due to Kupershmidt \cite{Kup85} (see \cite{DSKVW18-1} for a proof of this). However, each framework has its own advantages and disadvantages. Broadly speaking, using Poisson vertex algebras and the machinery of the so-called $\lambda$-brackets makes the computations more straightforward and easy to be implemented in a Computer Algebra System (CAS). For some of the computations required in this paper we adapted the Mathematica package \texttt{MasterPVA} \cite{casati16:_master_walg} to the multiplicative case. On the other hand, the $\theta$ formalism is more efficient in dealing with the complex of poly-vectors and their symmetries.

\subsection{Algebras of difference functions and vector fields}\label{sec:2.1}

Let us summarise the setting of the theory of Hamiltonian differential-difference equations (see \cite{DSKVW18-1}). We denote by $(\mathcal{P}_\ell,\cS)$
the algebra of \emph{difference polynomials} in $\ell$ variables $u=(u^1\,,\dots\,u^\ell)$.  By definition, this is the algebra of polynomials over the field $\C$ in the variables $\{u^i_n\}$, $i=1,\ldots,\ell$, $n\in\Z$ (we are denoting $u^i_0=u^i$),
endowed with an automorphism $\cS$,
called the \emph{shift operator}, which acts on generators as $\cS u^i_n=u^i_{n+1}$. A straightforward computation shows that $\cS$ satisfies
\begin{equation}\label{eq:shiftcomm}
\cS \frac{\dev}{\dev u^i_n}=\frac{\dev}{\dev u^i_{n+1}}\cS.
\end{equation}

An \emph{algebra of difference functions} in $\ell$ variables is a commutative associative unital algebra $\A$, containing $\mathcal{P}_\ell$, endowed with derivations $\frac{\dev}{\dev u^i_n}$ and an automorphism $\cS$, extending those in $\mathcal{P}_\ell$, satisfying \eqref{eq:shiftcomm} and $\frac{\dev f}{\dev u^i_n}=0$ for all but finitely many pairs of the indices $(i,n)$.

We denote by $\mc C=\{c\in\mc A\mid \cS c=c\}\subset\mc A$ the subalgebra of \emph{constants} and by
$$
\bar{\mc C}=\left\{f\in\mc A\,\middle|\,\frac{\partial f}{\partial u_n^i}=0 \text{ for all }i=1,\dots,\ell, n\in\mb Z\right\}\subset\mc A
$$
the subalgebra of \emph{quasiconstants}. It follows from \eqref{eq:shiftcomm} that $\mc C\subset \bar{\mc C}$ and that
$\cS(\bar{\mc C})=\bar{\mc C}$.

We define an ordering $\epsilon:\mb Z\to\mb N$ of the set of integers as follows:
$$
\epsilon(n)=\left\{
\begin{array}{cc}
2n\,, & n\geq0\,,
\\
-2n-1\,,& n<0
\,.
\end{array}
\right.
$$
Associated to this ordering, the algebra $ \mc A $ carries a filtration by subspaces ($n\in\mb Z$, $i=1,\dots,\ell$)
$$
 \mc A_{n,i} = \left\{ f \in \mc A \, \middle| \, \frac{\partial f}{\partial u^j_{m}} = 0
\,\text{ for } (\epsilon(m), j) > (\epsilon(n),i ) \, \text{in the lexicographical order} \right\}.
$$
For every $n\in \mb Z$ we let $\mc A_{n,0}=\mc A_{m,\ell}$, where $m\in\mb Z$ is such that $\epsilon(m)=\epsilon(n)-1$.
We also let $\mc A_{0,0}=\overline{\mc C}$. Hence, the first few terms of the filtration look as follows:
\begin{equation}\label{eq:filtration}
\begin{split}
&\mc C\subset\overline{\mc C}
\subset\mc A_{0,1}\subset\dots\subset\mc A_{0,\ell}=\mc A_{-1,0}
\subset\mc A_{-1,1}\subset\dots\subset\mc A_{-1,\ell}=\mc A_{1,0}
\\
&\subset\mc A_{1,1}\subset\dots\subset\mc A_{1,\ell}=\mc A_{-2,0}\subset\dots\subset\mc A_{n,i}
\subset\dots\subset\mc A\,.
\end{split}
\end{equation}
The algebra of difference functions $ \mc A $ is called \emph{normal}  if $ \frac{\partial}{\partial u^i_{n}} (\mc A_{n,i}) = \mc A_{n,i} $  for all $i=1,\dots,\ell, n \in \mb Z$.
\begin{remark}\label{rem:filtration}
The filtration \eqref{eq:filtration}, hence the definition of a normal algebra of difference functions, is slightly different from the filtration used in \cite{DSKVW18-1}. However, it can be checked that the statements of Theorem 5.1 and Theorem 5.2 in \cite{DSKVW18-1} still hold in our setting.
\end{remark}
\begin{remark}\label{rmk:A0-def}
Let us denote $\mc A_{0,\ell}=\mc A_{-1,0}$ as $\mc A_0=\{f\in\mc A\mid \frac{\partial f}{\partial u_n^s}=0, \text{ for every }s=1,\dots,\ell\,,n\neq0\}$. Then $\mc A_0\subset\mc A$ is a subalgebra with respect to the commutative associative product of $\mc A$.
\end{remark}
In a normal algebra of difference functions $\mc A$ we can "integrate" with respect to the variables $u_n^i$ in the following sense. For every $f\in\mc A_{n,i}$, since, by assumption, the partial derivative $ \frac{\partial}{\partial u^i_{n}}:\mc A_{n,i}\rightarrow \mc A_{n,i} $ is surjective, there exists $g\in\mc A_{n,i}$ such that $\frac{\partial g}{\partial u_{n}^i}=f$.
Hence, we denote
\begin{equation}\label{integral}
\int f \,\ud u_n^i=g\,.
\end{equation}
Clearly, we have
$\frac{\partial}{\partial u^i_{n}} (\int f\,  \ud u^i_n)=f$ and 
$\int f \, \ud u^i_{n}$ is uniquely defined up to adding an element of
$\ker \frac{\partial}{\partial u^i_{n}}|_{\mc A_{n,i}}=\mc A_{n,i-1}$.

The elements of the quotient space
$$
\F=\frac{\A}{(\cS-1)\A}
$$
are called \emph{local functionals}.
This name is motivated from the fact that the space $\A$ can be regarded as the space of densities of local functionals of fields
over a discrete lattice (see \cite{CW19} for details).
Hence, 
we denote the projection map $\A\twoheadrightarrow\F$ with the integral sign: $\int f=f+(\cS-1)\A$, $f\in\A$. We then clearly have
$\int f=\int \cS f$.

The \emph{variational derivative} of a local functional $F=\int f$ with respect to the variable $u^i$ is defined by
$$
\frac{\delta F}{\delta u^i}=\sum_{n\in\Z} \cS^{-n}\frac{\dev f}{\dev u^i_n}
\,.
$$
It follows form \eqref{eq:shiftcomm} that $\frac{\delta}{\delta {u^i}}(\cS-1)f=0$, for every $i=1\,,\dots\,,\ell$. Hence, we have a well defined map $\frac{\delta}{\delta u}:\F\to\mc A^\ell$ defined by
$$
\frac{\delta F}{\delta u}=\left(\frac{\delta F}{\delta u^i}\right)_{i=1}^\ell
\,,
$$
which we call the variational derivative operator.

A \emph{vector field} is, by definition, a derivation $X$ of $\mc A$ of the form
$$
X=\sum_{i=1}^\ell\sum_{n\in\mb Z}X_{i,n}\frac{\partial}{\partial u_n^i}
\,,
\quad X_{i,n}\in\mc A\,.
$$
We denote by $\Vect(\mc A)$ the Lie algebra of vector fields. An \emph{evolutionary vector field} is a vector field commuting with $\cS$ and we denote by $\Vect^{\mc S}(\mc A)\subset\Vect(\mc A)$ the Lie subalgebra of evolutionary vector fields.
By \eqref{eq:shiftcomm}, evolutionary vector fields are in one-to-one correspondence with elements $Q=(Q^i)_{i=1}^\ell\in\mc A^\ell$ via
\begin{equation}\label{eq:vf}
Q\mapsto X_Q=\sum_{i=1}^\ell\sum_{n\in\Z}\cS^n( Q^i)\frac{\dev}{\dev u^i_n}
\,.
\end{equation}
Conversely, given an evolutionary vector field $X$, we associate to it the vector $(X(u^i))_{i=1}^\ell\in\A^\ell$, which is called the \emph{characteristics} of $X$.
Given two evolutionary vector fields $X_P$ and $X_Q$ their Lie bracket is
\begin{equation}\label{eq:commvf}
[X_P,X_Q]=
X_{[P,Q]}\,,
\end{equation}
where ($i=1,\dots,\ell$)
$$
[P,Q]^i=X_P(Q^i)-X_Q(P^i)
=\sum_{j=1}^\ell\sum_{n\in\mb Z}\left(S^n(P^j)\frac{\partial Q^i}{\partial u_n^j}
-S^n(Q^j)\frac{\partial P^i}{\partial u_n^j}\right)
\,.$$
Moreover, given an evolutionary vector field $X$ of characteristic $Q\in\mc A^\ell$, we denote by $\hat{X}$ the induced map $\hat{X}\colon\F\to\F$ defined by
\begin{equation}\label{eq:fvfield-action}
\hat{X}(F)=\hat{X}\left(\int f\right)=\int X(f)=\int Q \cdot \frac{\delta F}{\delta u}\,,
\end{equation}
where the in last identity we are using the usual dot products of vectors.
We call $\hat{X}$ a \emph{local 1-vector}. Local 1-vectors form a Lie algebra with bracket induced by \eqref{eq:commvf}.

In the sequel, with a slight abuse of notation, we will drop the hat sign and we will use the same symbol for an evolutionary vector field of $\A$ and for its local 1-vector on $\F$.

\subsection{Difference operators and multiplicative Poisson vertex algebras}
\label{sec:2.2}
We denote by $\mc M_\ell(\mc A)=\Mat_{\ell\times\ell}(\mc A[\mc S,\mc S^{-1}])$ the \emph{algebra of (local) matrix difference operators}.
Elements of $\mc M_\ell(\mc A)$ are Laurent polynomials
\begin{equation}\label{eq:diffop-gen}
K(\mc S)=\sum_{l=M}^N K_{l}\mc S^l
\,,
\qquad
K_{l}=(K^{ij}_{l})_{i,j=1}^\ell\in\Mat_{\ell\times\ell}(\mc A)
\,,
\end{equation}
with the associative product $\circ$ defined by the relation
$$
\mc S\circ A=\mc S(A)\mc S
\,,\qquad A \in\Mat_{\ell\times\ell}(\mc A)\,.
$$
We say that $K(\mc S)$ as in \eqref{eq:diffop-gen} is a difference operator of order $(M,N)$.
A difference operator of this class defines a bilinear map $\F\times\F\to\F$ by the formula
\begin{equation}\label{eq:prebiv}
B(F,G)=\int\frac{\delta F}{\delta u}\cdot K(\mc S)\left(\frac{\delta G}{\delta u}\right)
\,.
\end{equation}
The \emph{adjoint} of a difference operator $K(\mc S)=(K^{ij}(\mc S))$ as in \eqref{eq:diffop-gen} is the difference operator
$K^{*}(\mc S)=((K^*)^{ij}(\mc S))_{i,j=1}^\ell$, where
\begin{equation}\label{eq:diffop-*}
(K^*)^{ij}(\mc S)=\sum_{l=M}^N \cS^{-l}\circ K^{ji}_{l}
\,.
\end{equation}
If $K$ is skewadjoint, namely $K^*(\mc S)=-K(\mc S)$, then the bilinear map $B$ is skewsymmetric.
Viceversa, if $B$ as in \eqref{eq:prebiv} is skewsymmetric, then $K(\mc S)$ is skewadjoint (such $B$ will be called a local 2-vector in Section \ref{sec:2.3}).
Note that, a local difference operator of the form $\eqref{eq:diffop-gen}$ is skewadjoint if and only if $M=-N$ for $N\geq0$ and
\begin{equation}\label{eq:diffop-skew}
K^{ij}_{-l}=-\cS^{-l} K^{ji}_{l},
\end{equation}
for every $i,j=1,\dots,\ell$ and $l=0,\dots,N$, \cite[Proposition 9]{CW19}.
In particular, this means that $K_{0}$ is a skewsymmetric matrix.

A \emph{(local) multiplicative $\lambda$-bracket} on an algebra of difference functions $\A$ is a bilinear operation $\A\times\A\to\A[\lambda,\lambda^{-1}]$
that
satisfies the following properties:
\begin{enumerate}
\item sesquilinearity: $\{\cS f_{\lambda}g\}=\lambda^{-1}\{f_{\lambda}g\}$ and $\{f_{\lambda}\cS g\}=\lambda\cS\{f_{\lambda}g\}$,
for every $f,g\in\A$;
\item left and right Leibniz rules: $\{f_{\lambda}gh\}=\{f_{\lambda}g\}h+g\{f_{\lambda}h\}$, $\{fg_{\lambda}h\}=\{f_{\lambda\cS}h\}g+\{g_{\lambda\cS}h\}f$, for every $f,g,h\in\A$.
\end{enumerate}
The notation $\{f_{\lambda\cS}g\}h$ should be read as follows:
write 
\begin{equation}\label{20240528:eq1}
\{f_{\lambda}g\}=\sum_{p=M}^N c_p \lambda^p
\,,
\end{equation}
then
$\{f_{\lambda\mc S}g\}h=\sum_{p=M}^N c_p \lambda^p\mc S^ph$,
namely we replace $\lambda$ with $\lambda\cS$ in the expansion of the $\lambda$-bracket and we let the shift operator act on the right.

Given a $\lambda$-bracket on $\mc A$ we define a matrix difference operator
\begin{equation}\label{20240918:eq1}
K(\mc S)=(\{u^j_{\mc S} u^i\})_{i,j=1}^\ell
\end{equation}
by replacing $\lambda$ with the shift operator $\mc S$. On the other hand, given a matrix difference operator $K(\mc S)$
we associate to it a multiplicative $\lambda$-bracket on $\mc A$ according to the following 
\emph{Master Formula} \cite{DSKVW18-1}
\begin{equation}\label{eq:masterf}
\{f_{\lambda}g\}_K=\sum_{i,j=1}^\ell\sum_{m,n\in\Z}\frac{\dev g}{\dev u^j_m}(\lambda\cS)^mK^{ji}(\lambda\mc S)(\lambda\cS)^{-n}\frac{\dev f}{\dev u^i_n}
\,.
\end{equation}
In particular, writing $K(\mc S)$ as in \eqref{eq:diffop-gen}, the Master Formula \eqref{eq:masterf} on generators becomes
\begin{equation}\label{eq:lambdadef}
\{u^i_\lambda u^j\}_K=K^{ji}(\lambda)=\sum_{l=M}^NK^{ji}_{l}\lambda^l
\,.
\end{equation}
When $\mc A=\mc P_\ell$, equations \eqref{20240918:eq1} and \eqref{eq:masterf} give a
one-to-one correspondence between multiplicative $\lambda$-brackets on $\mc P_\ell$ and the algebra of matrix difference operators $\mc M_{\ell}(\mc P_\ell)$ (this is also true for some algebras of difference functions, see Remark 2.3 in \cite{DSKVW18-1}).

Note also that the bilinear operation defined in \eqref{eq:prebiv} is given in terms of the multiplicative $\lambda$-bracket \eqref{eq:masterf}
as follows
\begin{equation}
B(F,G)=\int\{g_{\lambda}f\}_K\big|_{\lambda=1}.
\end{equation}

For an element $f\in\mc A$ we define its \emph{Fr\'echet derivative}
as the difference operator $D_f(\mc S):\mc A^\ell\rightarrow\mc A$ defined by (see \eqref{eq:vf})
\begin{equation}\label{eq:frechet}
D_f(\mc S)Q=X_Q(f)
=\sum_{i=1}^\ell\sum_{n\in\mb Z}\frac{\partial f}{\partial u_n^i}\mc S^nQ^i\,,
\quad Q=(Q^i)_{i=1}^\ell\in\mc A^\ell
\,.
\end{equation}
More generally, for a collection of elements $\phi=(\phi^1,\dots \phi^m)\in\mc A^m$, the corresponding Fr\'echet derivative is the
matrix 
difference operator $D_\phi(\mc S):\mc A^\ell\rightarrow\mc A^m$ defined by ($k=1,\dots,m$)
\begin{equation}\label{eq:frechet2}
\left(D_\phi(\mc S)Q\right)_k=D_{\phi^k}(\mc S)Q
\,\left(\,=X_Q(\phi^k)\right)
\,,
\quad Q=(Q^i)_{i=1}^\ell\in\mc A^\ell
\,.
\end{equation}
Explicitly, $D_\phi(\mc S)=\left(D_{\phi}^{ij}(\mc S)\right)\in\Mat_{m\times\ell}(\mc A[\mc S,\mc S^{-1}])$ is the matrix difference operator with entries
\begin{equation}\label{eq:frechet3}
D_\phi^{ij}(\mc S)
=\sum_{n\in\mb Z}\frac{\partial \phi^i}{\partial u_n^j}\mc S^n
\,,
\qquad
i=1,\dots,m\,,j=1\dots,\ell
\,.
\end{equation}
By equation \eqref{eq:diffop-*}, the adjoint of the Fr\'echet derivative operator $D_\phi(\mc S)$ is the matrix difference operator
$D_\phi^*(\mc S)=\left((D_{\phi}^*)^{ij}(\mc S)\right)\in\Mat_{\ell\times m}(\mc A[\mc S,\mc S^{-1}])$, where
\begin{equation}\label{eq:frechet4}
(D_\phi^*)^{ij}(\mc S)
=\sum_{n\in\mb Z}\mc S^{-n}\circ \frac{\partial \phi^j}{\partial u_n^i}
\,,
\qquad
i=1,\dots,m\,,j=1\dots,\ell
\,.
\end{equation}
In particular, from \eqref{eq:frechet4}, we have that
$D_f^*(\mc S)(1)=\frac{\delta F}{\delta u}$, where $F=\int f$, $f\in\mc A$.
Hence, using equations \eqref{eq:frechet} and \eqref{eq:frechet4} we can then rewrite the Master Formula \eqref{eq:masterf} as
\begin{equation}\label{eq:frechet5}
\{f_\lambda g\}_K=D_{g}(\lambda \mc S)K(\lambda \mc S)D_f^*(\lambda\mc S)(1)
\,.
\end{equation}

Let $\mc A$ be an algebra of difference functions
endowed with a multiplicative $\lambda$-bracket
$\{\cdot\,_\lambda\,\cdot\}$. 
As in \cite{DSKVW18-1},
we say that $\mc A$ is 
a \emph{multiplicative Poisson vertex algebra} (multiplicative PVA) if 
the multiplicative $\lambda$-bracket satisfies the skewsymmetry property
\begin{equation}\label{eq:pva-skew}
\{g_{\lambda}f\}=-{}_\to\{f_{(\lambda\cS)^{-1}}g\}
\end{equation}
and the Jacobi identity
\begin{equation}\label{eq:pva-jac}
\{f_{\lambda}\{g_{\mu}h\}\}-\{g_\mu\{f_\lambda h\}\}-\{\{f_{\lambda}g\}_{\lambda\mu}h\}=0
\,,\end{equation}
for every $f,g,h\in\A$.
The notation used in \eqref{eq:pva-skew} means the following: let $\{f_\lambda g\}$ be as in \eqref{20240528:eq1}, then
${}_\to\{f_{(\lambda\cS)^{-1}}g\}=\sum_{p=M}^N \left(\lambda\cS\right)^{-p}c_p$,
namely the formal parameter $\lambda$ of the bracket is replaced by $(\lambda\cS)^{-1}$ and the shift operator acts on the coefficients $c_p$.

Let $K(\mc S)$ be a matrix difference operator and $\{\cdot\,_\lambda\,\cdot\}_K$ be the corresponding $\lambda$-bracket
on $\mc A$ defined by \eqref{eq:masterf}. It is shown in \cite{DSKVW18-1} that skewsymmetry \eqref{eq:pva-skew} holds if and only if $K(\mc S)$ is skewadjoint, and that, assuming $K(\mc S)$ is skewadjoint, the Jacobi identity \eqref{eq:pva-jac} holds if and only if
\begin{equation}\label{eq:pva-jac-K}
\begin{split}
&
\sum_{s=1}^\ell\sum_{n\in\mb Z}\left(
\frac{\partial H^{kj}(\mu)}{\partial u_n^s}(\lambda \mc S)^nH^{si}(\lambda)
-\frac{\partial H^{ki}(\lambda)}{\partial u_n^s}(\mu \mc S)^nH^{sj}(\mu)
\right)
\\
&
=\sum_{s=1}^\ell\sum_{n\in\mb Z}
H^{ks}(\lambda\mu\mc S)(\lambda\mu\mc S)^{-n}\frac{\partial H^{ji}(\lambda)}{\partial u_n^s}
\,,
\quad\text{for every  $i,j,k=1\,\dots\,\ell$}
\,.
\end{split}
\end{equation}
We say that a local difference operator of the form $\eqref{eq:diffop-gen}$ is \emph{Hamiltonian} if and only if
its associated $\lambda$-bracket \eqref{eq:lambdadef} defines a multiplicative Poisson vertex algebra structure on $\A$, namely it is skewadjoint and the condition \eqref{eq:pva-jac-K} holds.

\subsection{Local poly-vectors and the \texorpdfstring{$\theta$}{theta} formalism}\label{sec:2.3}
A local 1-vector is a linear map $\F\to\F$ of the form \eqref{eq:fvfield-action}. In this paragraph we extend the definition of local $1$-vectors to local poly-vectors and recall the $\theta$ formalism as a convenient framework to describe their complex and their operations.

First, recall from \cite{CW19} that a local $p$-vector, in the context of difference algebras, is an alternating $p$-linear map $\F^p\to\F$ of the following form
\begin{equation}\label{eq:pvf}
B(F_1,\ldots,F_p)=\sum\int\frac{\delta F_1}{\delta u^{i_1}}K^{i_1,i_2,\ldots,i_p}_{n_2,\ldots,n_p}\left(\cS^{n_2}\frac{\delta F_2}{\delta u^{i_2}}\right)\left(\cS^{n_3}\frac{\delta F_3}{\delta u^{i_3}}\right)\cdots\left(\cS^{n_p}\frac{\delta F_p}{\delta u^{i_p}}\right).
\end{equation}
We identify local $0$-vectors with local functionals in the obvious way. Moreover, a skewadjoint difference operator defines a local $2$-vector (or \emph{bivector}) according to the formula \eqref{eq:prebiv}, since the bilinear map $B$ is skewsymmetric; conversely, any bivector is defined by a skewadjoint difference operator, see \cite{CW19}.
The Lie bracket of local 1-vector and the action of a local 1-vector on a local functional, i.e. the Lie derivative of a local functional along the local 1-vector, have already been described in \eqref{eq:commvf} and \eqref{eq:fvfield-action}.

In order to characterise the full complex of local poly-vectors (from now on we skip the term ``local" and simply call them poly-vectors), we adopt the so-called ``$\theta$ formalism'', which is now a standard machinery in the differential case (see \cite{G01,lz11,lz13,ccs1,ccs2}) and has already been adapted to the difference case (see \cite{CW19} for more details), which extends to the infinite dimensional setting the well-known fact that, for a finite dimensional manifold $M$,  $T^*[1]M$ is isomorphic to $\Gamma(\Lambda (TM))$.

To this aim, we introduce the space of \emph{densities of poly-vectors}
\begin{equation}
\hA=\A[\theta_{i,n}\mid i=1,\ldots,\ell,\;n\in\Z].
\end{equation}
The space $\hA$ is graded according to
\begin{align*}
\deg f&=0\,,\quad f\in\mc A\,, &\deg \theta_{i,n}&=1\,,
\end{align*}
and the product in $\hA$ is graded commutative: $fg=gf$, $f \theta_{i,n}=\theta_{i,n}f$, for $f,g\in\mc A$, but $\theta_{i,n}\theta_{j,m}=-\theta_{j,m}\theta_{i,n}$. The shift operator naturally extends to $\hA$ by $\cS\theta_{i,n}=\theta_{i,n+1}$ (we will often simply denote
$\theta_{i,0}=\theta_i$). 
We also extend the filtration \eqref{eq:filtration} to a filtration of $\hA$ by letting
$$
\hA^{p}_{n,i}=\left\{\omega\in \hA^p\middle| \, \frac{\partial \omega}{\partial u^j_{m}} = 0
\,\text{ for } (\epsilon(m), j) > (\epsilon(n),i ) \, \text{in the lexicographical order} \right\}
\,.
$$
for every $n\in\mb Z$, $i=1,\dots,\ell$ and $p\geq0$. For every $n\in \mb Z$ we also let $\hA_{n,0}^p=\hA_{m,\ell}^p$, where $m\in\mb Z$ is such that $\epsilon(m)=\epsilon(n)-1$. 

The main point of this construction is the identification of 
$$
\hF:=\frac{\hA}{(\cS-1)\hA}
$$
with the space of poly-vectors described above.
Since $\cS$ respects the gradation of $\hA$, the space $\hF$ is graded as well and its homogeneous components $\hF^p$ of degree $p$ are in one-to-one correspondence with $p$-vectors. Indeed, a generic element of $\hF^p$ is of the form
\begin{equation}\label{eq:pvect-theta}
B=\sum_{i_1,\ldots,i_n=1}^\ell\,\sum_{n_2,\ldots,n_p\in\Z}\int\theta_{i_1} K^{i_1,i_2,\ldots,i_p}_{n_2,\ldots,n_p}\theta_{i_2,n_2}\theta_{i_3,n_3}\cdots\theta_{i_p,n_p},
\end{equation}
$K^{i_1,\ldots,i_p}_{n_2,\ldots,n_p}\in\A$, and it is in 1-to-1 correspondence with the $p$-vector in \eqref{eq:pvf}.
%
%

On the space of poly-vectors it is defined a degree $-1$ bilinear operation called the \emph{Schouten bracket}, which extends the notion of commutator of 1-vectors and of Lie derivative. Such bracket enjoys of a graded version of the skewsymmetry and of the Jacobi identity, typical of a \emph{Gerstenhaber algebra}:
\begin{gather}\label{eq:sch-skew}
[B,A]=-(-1)^{(a-1)(b-1)}[A,B]\\ \label{eq:sch-jac}
[A,[B,C]]=[[A,B],C]+(-1)^{(a-1)(b-1)}[B,[A,C]]
\end{gather}
with $A\in\hF^a$ and $B\in\hF^b$. Note that this grading corresponds to the one of standard Schouten brackets and it is different from the (simpler) grading of the bracket defined in \cite{CW19} (which is the one commonly used in almost all the literature on the differential case, with the notable exception of Getzler's seminal paper \cite{G01}). Such bracket is nothing else that the canonical symplectic structure on (our replacement of) $T^*[1]M$ \cite{v02}, namely
\begin{equation}\label{eq:sch-def}
[A,B]=-\sum_{l=1}^\ell\int\left(\frac{\delta A}{\delta u^l}\frac{\delta B}{\delta \theta_l}+(-1)^a\frac{\delta A}{\delta \theta_l}\frac{\delta B}{\delta u^l}\right).
\end{equation}
The variational derivative with respect to $\theta$ in \eqref{eq:sch-def} and in the following is defined by the same formula as the variational derivative with respect to $u$; note that the partial derivative with respect to $\theta$ must be odd (i.e. $\dev_\theta(ab)=(\dev_\theta a)b+(-1)^{\deg a}a(\dev_\theta b)$) to be consistent with $\theta_{i,n}\theta_{j,m}=-\theta_{j,m}\theta_{i,n}$. Moreover, this choice of signs not only is consistent with properties \eqref{eq:sch-skew} and \eqref{eq:sch-jac}, but also coincides with the commutator induced by \eqref{eq:commvf} and the identification of $[X,R]=\mathcal{L}_X R$ for a 1-vector $X$ and any poly-vector $R$.

A skewadjoint difference operator $K=K(\mc S)$ as in \eqref{eq:diffop-gen} defines a bivector $P$ by
\begin{equation}\label{eq:bic-def}
P=\frac12\sum_{i,j=1}^\ell\int\theta_i K^{ij}(\theta_j)
=\frac{1}{2}\sum_{i,j=1}^\ell\sum_{l=M}^N\int\theta_i K^{ij}_{l}\theta_{j,l}.
\end{equation}
Note that, thanks to the skewsymmetry, we can recover the difference operator $K$ from $P$ by computing
\begin{equation}
\sum_{j=1}^\ell K^{ij}(\theta_j)=\frac{\delta P}{\delta \theta_i}.
\end{equation}

We say that a bivector is \emph{Poisson} if its Schouten bracket with itself vanishes $[P,P]=0$ (the identity is non trivial because of the graded skewsymmetry of the bracket). Moreover,
there is a 1-to-1 correspondence between Poisson bivectors and multiplicative PVAs structures on algebras of difference functions given by \eqref{eq:masterf}, see \cite{CW19} for further details.

\begin{remark}
Because of the interplay between skewsymmetry in the exchange of the variables $\theta_i$ and the integral operation, it is not always obvious how to impose the vanishing of an element of $\hF$. A convenient ``normalization operator'' has been introduced by Barakat in the differential setting \cite{b08}. Note that $B\in\hF^0$ is vanishing if and only if its variational derivative with respect to the variables $u^i$ is 0. Given an element $B$ of $\hF^p$ for $p\geq 1$, we choose as its representative in $\hA^p$ the element
\begin{equation}\label{eq:b-norm}
\mathcal{N}B=\sum_{a=1}^\ell \frac{1}{p}\theta_a\frac{\delta B}{\delta \theta_a}
\end{equation}
and observe that $B=0$ if and only if $\mathcal{N}B=0$.
\end{remark}

\subsection{Homomorphisms of algebras of difference functions and point transformations}
Let $\mc A$ be an algebra of difference functions in $\ell$ variables $u^1,\dots,u^\ell$, let $\mc B$ be an algebra of difference functions in $m$ variables $v^1,\dots v^m$.
By a homomorphism of algebras of difference functions we mean an associative algebra homomorphism $\phi:\mc B\to\mc A$
such that $\phi(\mc S f)=\mc S\phi(f)$, for every $f\in\mc B$ (by an abuse of notation we are denoting the shift automorphisms of $\mc A$ and $\mc B$ with the same symbol).
Assuming that $\phi$ is injective, we have $m\leq\ell$ (in the examples that will interest us we have $m=\ell$) and we can identify $\mc B$ with the subalgebra $\phi(\mc B)\subset\mc A$.
We can regard injective homomorphisms of algebras of difference functions as  ``changes of variables". In this section we then review the usual change of variables formulas for vector fields, variational derivatives and difference operators, and we describe a special class of homomorphisms, that are called point transformations, which will be used throughout the paper.

Denote by $\phi^i=\phi(v^i)\in\mc A$, for every $i=1,\dots,m$, and consider 
the Fr\'echet derivative operator $D_\phi(\mc S)=\left(D_{\phi}^{ij}(\mc S)\right)\in\Mat_{m\times\ell}(\mc A[\mc S,\mc S^{-1}])$ associated to the collection
$(\phi^1,\dots,\phi^m)\in\mc A^m$ using \eqref{eq:frechet4}.
Let also $K(\mc S)\in \mc M_\ell(\mc A)$ be a matrix-difference operator and consider the corresponding $\lambda$-bracket on $\mc A$ given by the Master Formula \eqref{eq:masterf} or equivalently by \eqref{eq:frechet5}. 
From equation \eqref{eq:frechet},
we have $D_{u^i}(\mc S)=(\delta_{ij})_{j=1}^\ell$, for every $i=1,\dots,\ell$.
Hence, from the Master Formula \eqref{eq:frechet5} we get ($f\in\mc B$)
\begin{equation}\label{tr1}
\{{u^i}{}_{\lambda}\phi(f)\}_K
=\left(D_{\phi(f)}(\lambda \mc S)K(\lambda)\right)_i
\,,
\qquad i=1.\dots,\ell
\,.
\end{equation}
On the other hand, using the Master Formula \eqref{eq:masterf} twice, first with respect to the variables $\phi^1=\phi(v^1),\dots,\phi^m=\phi(v^m)$, and then with respect to the variables $u^1,\dots,u^\ell$, we get ($i=1,\dots,\ell$)
\begin{equation}\label{tr2}
\begin{split}
&\{{u^i}_\lambda \phi(f)\}_K
=\sum_{k=1}^m\sum_{n\in\mb Z}\phi\left(\frac{\partial f}{\partial v^k_n}\right)(\lambda\mc S)^n\{u^i{}_\lambda \phi^k\}_K
\\
&=\sum_{j=1}^\ell\sum_{k=1}^m\sum_{n,\tilde n\in\mb Z}
\phi\left(\frac{\partial f}{\partial v^k_n}\right)(\lambda\mc S)^n
\frac{\partial \phi^k}{\partial u^j_{\tilde n}}(\lambda\mc S)^{\tilde n}
\{u^i{}_\lambda u^j\}_K
\\
&
=\left(\phi\left(D_{f}(\lambda\mc S)\right)
D_{\phi}(\lambda\mc S)K(\lambda)
\right)_i
\,.
\end{split}
\end{equation}
In the last equality of \eqref{tr1} we used \eqref{eq:frechet} and
\eqref{eq:frechet3} and we note that $D_{f}(\mc S)$ is the Fr\'echet derivative \eqref{eq:frechet} of $f\in\mc B$ which computed with respect to the variables $v^1,\dots,v^m$.
Combining equations \eqref{tr1} and \eqref{tr2} we get the identity
$$
\left(
D_{\phi(f)}(\lambda S)-\phi\left(D_{f}(\lambda\mc S)\right)
D_{\phi}(\lambda\mc S)
\right)K(\lambda)=0\,,
$$
which holds for any matrix difference operator $K(\mc S)\in\mc M_\ell(\mc A)$.
We then get, for every $f\in\mc B$, 
\begin{equation}\label{tr3}
D_{\phi(f)}(\lambda)
=\phi\left(D_{f}(\lambda\mc S)\right)
D_{\phi}(\lambda)
\,.
\end{equation}
Let $X_Q$ be an evolutionary vector field on $\mc A$ as in \eqref{eq:vf} such that
$X_Q(\mc B)\subset\mc B$ (as before we are identifying $\mc B\simeq\phi(\mc B)\subset\mc A$). Using equations \eqref{eq:frechet} and \eqref{eq:frechet2}
we have that $X_Q(\phi^i)=D_{\phi^i}(\mc S)Q=(D_{\phi}(\mc S)Q)_i\in\mc B$, for every $i=1,\dots,m$.
Hence, applying both sides of \eqref{tr3} to $Q$ and using \eqref{eq:frechet} we get the identity
$$
X_Q(\phi(f))=D_{\phi(f)}(\mc S)Q
=\phi\left(D_{f}(\mc S)\right)
D_{\phi}(\mc S)Q
=\phi\left(X_{D_\phi(\mc S)Q}(f)\right)
\,,
$$
which states the change of variables formula for the restriction of the vector field $X_Q$ to $\mc B$ in terms of the vector field $X_{D_{\phi}(\mc S)Q}$
on $\mc B$.

Let $F=\int f\in\mc B/(\mc S-1)\mc B$. Since $\phi$ commutes with $\mc S$ we denote
by $\phi(F)=\phi(\int f)=\int \phi(f)\in\mc A/(\mc S-1)\mc A$.
Recall from Section \ref{sec:2.2} that $\frac{\delta \phi(F)}{\delta u}=D_{\phi(f)}^*(\mc S)(1)\in\mc A^\ell$ and similarly $\frac{\delta F}{\delta v}=D_{f}^*(\mc S)(1)\in\mc B^m$ (note that the latter Fr\'echet derivative is computed with respect to the variables $v^1,\dots,v^m$).
Computing the adjoint operators of both sides of \eqref{tr3} and applying them to $1$ we get the change of variables relation between the variational derivative of $F$ and $\phi(F)$:
$$
\frac{\delta \phi(F)}{\delta u}=
D_\phi^*(\mc S)\left(
\phi\left(\frac{\delta F}{\delta v}\right)
\right)
\,.
$$
Furthermore, let us denote by $\tilde K(\mc S)\in\mc M_{m}(\mc A)$ the matrix difference operator with entries $\tilde K^{ij}(\lambda)=\{\phi^j{}_\lambda\phi^i\}_K$, $i,j=1,\dots,m$. 
It immediately follows from \eqref{eq:frechet5} that
we have the change of variables relation between $K(\mc S$) and $\tilde K(\mc S)$
given by
\begin{equation}\label{tr5}
\tilde K(\lambda)=D_{\phi}(\lambda \mc S)K(\lambda \mc S)D^*_{\phi}(\lambda)
\,.
\end{equation}
If the entries of $\tilde K(\mc S)$ belong to $\mc B[\mc S,\mc S^{-1}]$ and $K(\mc S)$ is a Poisson structure on $\mc A$, then $\tilde K(\mc S)$ is a Poisson structure on $\mc B$. 

In \cite{Dub89,DN89}, the authors think of $(u^1,\dots,u^\ell)$
as local coordinates on an $\ell$-dimensional manifold $M$
and they classify a certain class of Hamiltonian difference operators up to local change of coordinates on $M$. The importance of considering invariance with respect to local change of coordinates of the underlying manifold $M$ stems from the fact that the invariant Hamiltonian difference operators provide interesting
"geometric objects" related to the manifold $M$, see for example Remark \ref{rem:DP}.
In terms of algebras of difference functions, a local change of coordinates correspond to a \emph{point transformation}.
\begin{definition}\label{def:pt}
Let $\mc A$ be an algebra of difference functions in the variables $u^1,\dots,u^\ell$ and let $\mc B$ be an algebra of difference functions in the variables $v^1,\dots,v^\ell$. An injective homomorphism of algebras of difference functions $\phi:\mc B\to\mc A$ is called a point transformation if $\phi(\mc B_0)\subset\mc A_0$.
\end{definition}
Let $\phi:\mc B\to\mc A$ be a point transformation and let, as before,
$\phi^i=\phi(v^i)\in\mc A$, $i=1,\dots,\ell$. In this case, since
by assumption $\phi^i\in\mc A_0$, we have $\frac{\partial \phi^i}{\partial u^{j}_n}=0$, for every $n\neq0$ (see Remark \ref{rmk:A0-def}) and  the Fr\'echet derivative 
\eqref{eq:frechet3} reduces to the Jacobian matrix
$$
J_\phi=D_\phi(\mc S)
=\left(\frac{\dev\phi^i}{\dev u^j}\right)_{i,j=1}^\ell\in\Mat_{\ell\times\ell}(\mc A)
\,.
$$
This matrix is non-degenerate (meaning that $\det J_{\phi}\neq0$) since $\phi$ is injective. 
Under a point transformation, the change of variables formula \eqref{tr5}
for a Hamiltonian operator $K(\mc S)$ becomes
\begin{equation}\label{eq:tildeK}
\tilde K(\mc S)=J_{\phi}K(\mc S)J_{\phi}^T
\,.
\end{equation}
Since $J_{\phi}$ has order 0 as a matrix difference operator and it is nondegenerate, then the matrix difference operators $K(\mc S)$ and $\tilde K(\mc S)$ have the same order.

An example of a point transformation is the homomorphism of algebras of difference functions
$\phi:\mb C[\tilde u_n^{\pm1},\tilde v_n^{\pm1}\mid n\in\mb Z]\rightarrow
\mb C[u_n^{\pm1},v_n^{\pm1},\log u_n\mid n\in\mb Z]$ defined on generators by
$$
\phi(\tilde u)=\log u\,,
\quad \phi(\tilde v)=v\,,
$$
used in Remark \ref{remark:Toda} to pass from the first Hamiltonian operator of the Toda lattice to the Hamiltonian operator \eqref{eq:intro1}.

\begin{remark}
If we drop the request that $\phi(\mc B_0)\subset\mc A_0$ in Definition \ref{def:pt}
we get a big group of transformations acting on the space of matrix Hamiltonian
difference operators extensively studied in \cite{cmw19,cmw20}.
It is clear from \eqref{tr5} that these transformations do not preserve the order of Hamiltonian operators, and invariant Hamiltonian operators under such transformations lose track of the geometric properties of the underlying manifold $M$. 

In fact,
since we can always embed an algebra of difference functions in a suitable completion of an algebra of differential functions (see for example \cite{Guido} or the appendix in \cite{VY23}),
an injective homomorphism between algebras of difference functions corresponds (under a mild non-degeneracy assumption) to a Miura-type transformation, in the sense of Dubrovin-Zhang \cite{DZ01}, between algebras of differential functions.
For example, the homomorphism $\phi:\mb C[v_n\mid n\in \mb Z]\rightarrow\mb C[u_n\mid n\in\mb Z]$
given on generators by $\phi(v_n)=u_nu_{n+1}$, which sends the modified Volterra lattice in the Volterra lattice \cite{kmw13}, gives rise to the Miura-type transformation
$$
\mu(v)=\sum_{k\in\mb Z_+}\frac{1}{k!}u\partial_x^ku\epsilon^k
=u^2+u\partial_xu\epsilon+\frac{1}{2}u\partial_x^2u\epsilon^2+\dots\,.
$$
The requirement $\phi(\mc B_0)\subset\mc A_0$ in Definition \ref{def:pt}
of a point transformation corresponds to the strongest condition that
the corresponding Miura transformation is independent of $\epsilon$.
\end{remark}

\section{Hamiltonian matrix difference operators}\label{sec:ham}
Recall from Section \ref{sec2}, that a matrix difference operator as in \eqref{eq:diffop-gen} is said to be \emph{Hamiltonian} if its associated bivector \eqref{eq:bic-def} is Poisson; alternatively, if its
associated $\lambda$-bracket \eqref{eq:lambdadef} defines a multiplicative Poisson vertex algebra.

In this paper, we are particularly interested in Hamiltonian operators of order $(-1,1)$; indeed, several examples of Hamiltonian structures for integrable difference equation belong to this class (see e.g. \cite{kmw13}), which can be regarded as a discretisation of the Hamiltonian operators of hydrodynamic type (namely, differential operators of order 1), see \cite{DN83,Dub89}.

A generic skewsymmetric difference operator $K=K(\mc S)$ of this class is of the form
\begin{equation}\label{eq:diffop-1st}
K=A\,\cS-\cS^{-1}\circ A^T+B,
\end{equation}
where $A,B\in\mathrm{Mat}_{\ell\times\ell}(\A)$ and $B=-B^T$.
The corresponding $\lambda$-bracket and bivector (in the $\theta$ formalism) are, respectively,
($1\leq i,j\leq \ell$)
\begin{align}\label{eq:lambda-1}
\{u^i_{\lambda}u^j\}_K&=A^{ji}\lambda-\cS^{-1}A^{ij}\lambda^{-1}+B^{ji},\\
P_K&=\int\left(A^{ij}\theta_i\theta_{j,1}+\tfrac12B^{ij}\theta_i\theta_j\right).
\end{align}
Substituting \eqref{eq:diffop-1st} into \eqref{eq:pva-jac-K} we have that the $\lambda$-bracket defined by \eqref{eq:lambda-1} satisfies
the Jacobi identity \eqref{eq:pva-jac} if and only if the following identity holds
for every $i,j,k=1,\dots,\ell$:
\begin{equation}
\label{20240617:eq1}
\begin{split}
&\sum_{s=1}^\ell\sum_{n\in\mb Z}
\left(
\frac{\partial A^{kj}}{\partial u_n^s}
\mc S^{n}A^{si}
+
\frac{\partial A^{kj}}{\partial u_{n+1}^s}
\mc S^{n+1}B^{si}
-
\frac{\partial A^{kj}}{\partial u_{n+2}^s}
\mc S^{n+1}A^{is}
\right)\lambda^{n+1}\mu
\\
&+\sum_{s=1}^\ell\sum_{n\in\mb Z}
\left(
\frac{\partial B^{kj}}{\partial u_n^s}
\mc S^{n}A^{si}
+
\frac{\partial B^{kj}}{\partial u_{n+1}^s}
\mc S^{n+1}B^{si}
-
\frac{\partial B^{kj}}{\partial u_{n+2}^s}
\mc S^{n+1}A^{is}
\right)\lambda^{n+1}
\\
&-\sum_{s=1}^\ell\sum_{n\in\mb Z}
\mc S^{-1}\left(
\frac{\partial A^{jk}}{\partial u_{n}^s}
\mc S^{n}A^{si}
+
\frac{\partial A^{jk}}{\partial u_{n+1}^s}
\mc S^{n+1}B^{si}
-
\frac{\partial A^{jk}}{\partial u_{n+2}^s}
\mc S^{n+1}A^{is}
\right)\lambda^{n}\mu^{-1}
\\
&-\sum_{s=1}^\ell\sum_{n\in\mb Z}
\left(
\frac{\partial A^{ki}}{\partial u_n^s}
\mc S^{n}A^{sj}
+
\frac{\partial A^{ki}}{\partial u_{n+1}^s}
\mc S^{n+1}B^{sj}
-
\frac{\partial A^{ki}}{\partial u_{n+2}^s}
\mc S^{n+1}A^{js}
\right)\lambda\mu^{n+1}
\\
&-\sum_{s=1}^\ell\sum_{n\in\mb Z}
\left(
\frac{\partial B^{ki}}{\partial u_n^s}
\mc S^{n}A^{sj}
+
\frac{\partial B^{ki}}{\partial u_{n+1}^s}
\mc S^{n+1}B^{sj}
-
\frac{\partial B^{ki}}{\partial u_{n+2}^s}
\mc S^{n+1}A^{js}
\right)\mu^{n+1}
\\
&+\sum_{s=1}^\ell\sum_{n\in\mb Z}
\mc S^{-1}\left(
\frac{\partial A^{ik}}{\partial u_{n}^s}
\mc S^{n}A^{sj}
+
\frac{\partial A^{ik}}{\partial u_{n+1}^s}
\mc S^{n+1}B^{sj}
-
\frac{\partial A^{ik}}{\partial u_{n+2}^s}
\mc S^{n+1}A^{js}
\right)\lambda^{-1}\mu^{n}
\\
&=
\sum_{s=1}^\ell\sum_{n\in\mb Z}
\left(
A^{ks}\mc S^n\frac{\partial A^{ji}}{\partial u_{1-n}^s}
+
B^{ks}\mc S^n\frac{\partial A^{ji}}{\partial u_{-n}^s}
-
(\mc S^{-1}A^{sk})\mc S^n\frac{\partial A^{ji}}{\partial u_{-1-n}^s}
\right)\lambda^{n+1}\mu^{n}
\\
&+
\sum_{s=1}^\ell\sum_{n\in\mb Z}
\left(
A^{ks}\mc S^n\frac{\partial B^{ji}}{\partial u_{1-n}^s}
+
B^{ks}\mc S^n\frac{\partial B^{ji}}{\partial u_{-n}^s}
-
(\mc S^{-1}A^{sk})\mc S^n\frac{\partial B^{ji}}{\partial u_{-1-n}^s}
\right)\lambda^{n}\mu^{n}
\\
&-\sum_{s=1}^\ell\sum_{n\in\mb Z}
\left(
A^{ks}\mc S^{n}\frac{\partial A^{ij}}{\partial u_{1-n}^s}
+
B^{ks}\mc S^{n}\frac{\partial A^{ij}}{\partial u_{-n}^s}
-
(\mc S^{-1}A^{sk})\mc S^{n}\frac{\partial A^{ij}}{\partial u_{-1-n}^s}
\right)\lambda^{n}\mu^{n+1}\,.
\end{split}
\end{equation}
\begin{proposition}\label{prop:dep}
Let $K$ be a Hamiltonian difference operator of the form \eqref{eq:diffop-1st}. If $A$ is nondegenerate, then
\begin{equation}\label{eq:cond0}
A^{ij}=A^{ij}(u,u_1),\qquad\qquad B^{ij}=B^{ij}(u)
\,,
\end{equation}
for every $i,j=1,\dots,\ell$ and denoting $u_q=(u^1_q,\ldots,u^\ell_q)$.
\end{proposition}
\begin{proof}
Since $K$ is Hamiltonian, equation \eqref{20240617:eq1} holds for every $i,j,k=1,\dots,\ell$.
Let us denote
\begin{align}\label{eq:pf1}
i_1&=\max\left\{i\,\Big|\;\exists \,l,m,n \text{ s.t. }\frac{\dev A^{lm}}{\dev u^n_i}\neq 0\right\}\\
i_0&=\max\left\{i\,\Big|\;\exists \,l,m,n \text{ s.t. }\frac{\dev B^{lm}}{\dev u^n_i}\neq 0\right\}\\\label{eq:pf2}
j_1&=\max\left\{j\,\Big|\;\exists \,l,m,n \text{ s.t. }\frac{\dev A^{lm}}{\dev u^n_{-j}}\neq 0\right\}\\
j_0&=\max\left\{j\,\Big|\;\exists \,l,m,n \text{ s.t. }\frac{\dev B^{lm}}{\dev u^n_{-j}}\neq 0\right\}.
\end{align}
First, let us suppose $i_0\geq 1$ and $i_1\geq2$ and consider the coefficient of $\lambda^{i_1+1}\mu$ in 
equation \eqref{20240617:eq1}. We obtain the identity
$$
\sum_{s=1}^\ell
\left(\frac{\dev A^{kj}}{\dev u^s_{i_1}}\cS^{i_1}A^{si}
+\frac{\dev A^{kj}}{\dev u^s_{i_1+1}}\cS^{i_1+1}B^{si}
-\frac{\dev A^{kj}}{\dev u^s_{i_1+2}}\cS^{i_1+1}A^{is}
\right)=0
\,,
$$
which, taken into account \eqref{eq:pf1}, reduces to
\begin{equation}\label{20240617:eq2}
\sum_{s=1}^\ell
\frac{\dev A^{kj}}{\dev u^s_{i_1}}\cS^{i_1}A^{si}=0
\,,
\end{equation}
for all $i, j, k=1,\dots,\ell$. If $A$ is nondegenerate, then \eqref{20240617:eq2} contradicts the definition of $i_1$ for $i_1\geq 2$, so we can conclude that $i_1\leq 1$. On the other hand, let us assume that $j_1\geq1$.
By computing the coefficient of $\lambda^{-j_1-2}\mu^{-1}$ in \eqref{20240617:eq1} we obtain the identity 
$$
\sum_{s=1}^\ell \frac{\dev A^{jk}}{\dev u^s_{-j_1}}\cS^{-j_1-1}A^{is}=0\,,
$$
which, if $A$ is nondegenerate, contradicts \eqref{eq:pf2}. Hence, we have that $j_1\leq 0$. Together with the constraint on $i_1$ this gives the first part of \eqref{eq:cond0}. 

Similarly, using the assumption that $A$ is nondegenerate, the same analysis for the coefficient of $\lambda^{i_0+1}\mu^0$ in \eqref{20240617:eq1} leads us to conclude that $i_0\leq 0$, while the same argument applied to the coefficient of $\lambda^{-j_0-1}\mu^0$ forces us to exclude $j_0\geq1$, giving the second part of \eqref{eq:cond0}.
\end{proof}
\begin{theorem}\label{thm:1st-ham}
Let $K$ be a difference operator of the form \eqref{eq:diffop-1st} such that
the conditions given by \eqref{eq:cond0} hold.
Then, $K$ is Hamiltonian if and only if
\begin{subequations}\label{eq:cond-comp}
\begin{gather}\label{eq:cond1}
\sum_{s=1}^\ell\left(
\frac{\dev A^{kj}}{\dev u^s_1} \cS A^{si}
-A^{ks}\cS\frac{\dev A^{ji}}{\dev u^s}
\right)
=0\,,\\
\label{eq:cond2}
\sum_{s=1}^\ell
\left(
\frac{\dev A^{kj}}{\dev u^s}A^{si}
+\frac{\dev A^{kj}}{\dev u^s_1}\cS B^{si}
-\frac{\dev A^{ki}}{\dev u^s}A^{sj}
-\frac{\dev A^{ki}}{\dev u^s_1}\cS B^{sj}
-A^{ks}\cS\frac{\dev B^{ji}}{\dev u^s}
\right)
=0\,,
\\
\label{eq:cond3}
\sum_{s=1}^\ell
\left(\frac{\dev A^{kj}}{\dev u^s_1}A^{is}
+\frac{\dev A^{kj}}{\dev u^s}B^{is}
-\frac{\dev A^{ij}}{\dev u^s_1}A^{ks}
-\frac{\dev A^{ij}}{\dev u^s}B^{ks}
+ \frac{\dev B^{ki}}{\dev u^s}A^{sj}\right)=0\,,\\
\label{eq:cond4}
\sum_{s=1}^\ell
\left(
\frac{\dev B^{kj}}{\dev u^s}B^{si}+\frac{\dev B^{ik}}{\dev u^s}B^{sj}+\frac{\dev B^{ji}}{\dev u^s}B^{sk}
\right)=0,
\end{gather}
\end{subequations}
for every $i,j,k=1,\dots,\ell$.
\end{theorem}
\begin{proof}
Using the assumption \eqref{eq:cond0}, by a straightforward computation, we have that the coefficients of 
$(\lambda^2\mu)^{\pm1}$, $(\lambda\mu^2)^{\pm1}$ and $(\lambda\mu^{-1})^{\pm1}$ vanish if and only if condition \eqref{eq:cond1} holds,
the coefficients of $\lambda\mu$, $\lambda^{-1}$ and $\mu^{-1}$ vanish if and only if condition \eqref{eq:cond2} holds,
the coefficients of $\lambda^{-1}\mu^{-1}$, $\lambda$ and $\mu$ vanish if and only if condition \eqref{eq:cond3} holds,
and the constant term (in $\lambda$ and $\mu$) vanishes if and only if condition \eqref{eq:cond4} holds.
\end{proof}
\begin{remark}\label{rmk:A0}
Consider $\mc \A_0$ of Remark~\ref{rmk:A0-def}. For every $f,g\in\mc \A_0$, let
\begin{equation}\label{20240618:eq1}
\{f,g\}=\sum_{i,j=1}^\ell \frac{\partial g}{\partial u^j}B^{ji}(u)\frac{\partial f}{\partial u^i}\in\mc A_0
\,.
\end{equation}
By Theorem \ref{thm:1st-ham}, $\mc A_0$ with the bracket \eqref{20240618:eq1} is a Poisson algebra.
\end{remark}
\begin{remark}\label{rem:DP}
As stated in \cite{Dub89} (see \cite{Parodi} for further details), in the case when the matrix $A$ is nondegenerate, the Hamiltonian operators of the form \eqref{eq:diffop-1st} are in one-to-one correspondence with admissible Poisson-Lie groups $(G,\rmat,\kmat)$ (here,
$\rmat\in\mathfrak g\otimes \mathfrak g$ and $\kmat\in\bigwedge^2\mathfrak g$, $\mathfrak g=\textrm{Lie}(G)$, are required to satisfy certain axioms).
Under this correspondence the matrices $A$ and $B$ take the form 
$$
A(u^l,u_1^l)=L(u^l) \rmat R(u_1^l)\,,
\qquad
B(u^l)=L(u^l) (\pi_G+\textrm{Ad}^{(2)}_{u^{-1}}\kmat) L(u^l)\,,
$$
where 
the columns of $L(u^l)$ (respectively, the rows of $R(u^l)$) define a basis of left invariant (respectively, right invariant) vector
fields on $G$ in a local chart $(u^1,\dots,u^\ell)$, $\pi_G$ is the Poisson structure of $G$, and $\textrm{Ad}^{(2)}$ denotes the adjoint action of $G$ on $\mathfrak g\otimes \mathfrak g$.
\end{remark}
\begin{remark}
When the matrix $A$ is nondegenerate, hence invertible, similarly to the proof of Theorem 2.5 of \cite{DSKVW18-1}, we can rewrite \eqref{eq:cond1} as
\begin{equation}\label{eq:cond1-bis}
\sum_{s=1}^\ell
\left(A^{-1}\right)_{is}\frac{ \dev A^{sj}}{\dev u^k_1}
=\sum_{s=1}^\ell
\cS\left(\frac{\dev A^{js}}{\dev u^i}\left(A^{-1}\right)_{sk}\right)
\,,
\end{equation}
for every $i,j,k=1,\dots,\ell$. However, in contrast with the scalar case ($\ell=1$), we cannot infer from \eqref{eq:cond1-bis} that $\dev_{u^k_1}\log A^{ij}=\cS\dev_{u^i}\log A^{jk}$, because $A$ and its derivatives do not necessarily commute (in particular, they \emph{do not} commute in the examples provided in \cite{kmw13}).
\end{remark}

\subsection{Normal forms for \texorpdfstring{$\ell=2$}{l=2} difference Hamiltonian operators}\label{sec:3.1}
In the following part of the paper, we focus on the $\ell=2$ component case. This is motivated by the richness of examples available in the literature (see \cite{kmw13} for a review) and in particular by the Toda lattice integrable system.

Consider a matrix difference operator $K\in\mc M_2(\mc A)$ as in \eqref{eq:diffop-1st}.
Throughout this section we assume that $\mc A$ has no zero-divisors and that $\mc C=\bar{\mc C}$ is a field. 

If $K$ is Hamiltonian, then Proposition \ref{prop:dep} fixes the 
dependency of $A$ and $B$ on the variables $u_q$, $q\in\mb Z$, when $A$ is nondegenerate. The classification for this class of 
operators has been studied in \cite{Dub89}. We present this result, for $\ell=2$, in Section \ref{sec:aff} in the non-Abelian case, and in Section  
\ref{sec:ab} in the Abelian case. For a degenerate matrix $A$, the dependence on the variables $u_1$ (to ensure that $K$ is Hamiltonian) provided by Proposition
 \ref{prop:dep} is not necessary. However, all the examples known to us (see \cite{kmw13}) satisfy it, 
therefore this is the case that we study comprehensively in Section \ref{sec:deg}. Finally, we present a few sparse examples of difference Hamiltonian operators with 
degenerate leading terms with more general dependency on the variables than the one required by Proposition \ref{prop:dep}.

\subsubsection{The Hamiltonian difference operator corresponding to \texorpdfstring{$\aff(\mb R)$}{Aff(R)}}\label{sec:aff}
As it is showed in \cite{Dub89}, this is the normal form -- depending on several parameters -- of non-constant Hamiltonian operators
$K$ as in \eqref{eq:diffop-1st} with nondegenerate leading term.

Let $G=\aff(\mb R)$ be the Lie group of affine transformations of $\mb R$. It is a two-dimensional non-Abelian
Lie group that we can represent as the matrix group
$$
G=\left\{\begin{pmatrix}u&v\\0&1\end{pmatrix}\Big| u,v\in\mb R, u\neq0\right\}\subset GL_2(\mb R)
\,.
$$
Its corresponding Lie algebra is the two-dimensional non-Abelian Lie algebra
$$
\mathfrak g=\left\{\begin{pmatrix}u&v\\0&0\end{pmatrix}\Big| u,v\in\mb R\right\}\subset gl_2(\mb R)
\,.
$$
We denote by $e_1=\begin{pmatrix}1&0\\0&0\end{pmatrix}$ and $e_2=\begin{pmatrix}0&1\\0&0\end{pmatrix}$ a choice for a basis of $\mf g$.
A basis of the space of left-invariant (resp. right-invariant) vector fields on $G$ is given by
$$
L_1=u\frac{\partial}{\partial u}\,,
\quad
L_2=u\frac{\partial}{\partial v}\,.
\quad
(\text{resp. } R_1=u\frac{\partial}{\partial u}+v\frac{\partial}{\partial v}\,,
\quad R_2=\frac{\partial}{\partial v}\,.
)
$$
Let $\rmat=\begin{pmatrix}a& b\\c& d\end{pmatrix}\in gl_2(\mb R)$
(note that, as vector spaces, $\mf g\otimes\mf g\cong gl_2(\mb R)$), then the matrix $A=A(u,v,u_1,v_1)$ from Remark \ref{rem:DP}
takes the form
\begin{equation}\label{eq:Aaff}
A=\begin{pmatrix}u&0\\0&u\end{pmatrix}
\begin{pmatrix}a&b\\c&d\end{pmatrix}
\begin{pmatrix}u_1&v_1\\0&1\end{pmatrix}
=
\begin{pmatrix}auu_1&u(a v_1+b)\\cuu_1&u(cv_1+d)\end{pmatrix}
\,.
\end{equation}
Consider also the Lie algebra $\mf g^*=\mb Re_1^*\oplus \mb Re_2^*$ with Lie bracket 
\begin{equation}\label{eq:20240621-gstar}
[e_1^*,e_2^*]=\alpha e_1^*+\beta e_2^*,\quad \alpha,\beta\in\mb R.
\end{equation} The corresponding bivector $\pi_G$ takes the form
$$
\pi_G=\begin{pmatrix}
0 & -\alpha(1-\frac1u)-\beta\frac vu\\\alpha(1-\frac1u)+\beta\frac vu&0
\end{pmatrix}
\,.
$$
Let $\kmat=\begin{pmatrix}0&\gamma\\-\gamma&0\end{pmatrix}\in so_2(\mb R)$
(note that, as vector spaces, $\bigwedge^2\mf g\cong so_2(\mb R)$). Then, the matrix $B=B(u,v)$ in Remark \ref{rem:DP}
takes the form
\begin{align}\notag
B&=\begin{pmatrix}u&0\\0&u\end{pmatrix}
\begin{pmatrix}
0 & -\alpha(1-\frac1u)-\beta\frac vu-\frac{\gamma}{u}\\\alpha(1-\frac1u)+\beta\frac vu+\frac{\gamma}{u}&0
\end{pmatrix}
\begin{pmatrix}u&0\\0&u\end{pmatrix}
\\\label{eq:Baff}
&=
\begin{pmatrix}
0 & -\alpha(u^2-u)-\beta uv-\gamma u\\\alpha(u^2-u)+\beta uv+ \gamma u&0
\end{pmatrix}
\,.
\end{align}
By Theorem \ref{thm:1st-ham}, the difference operator $K=A\mc S+B-\mc S^{-1}\circ A^T$ is Hamiltonian if and only if
the following conditions are met
\begin{equation}\label{20240618:eq2}
\left\{\begin{array}{l}
a\alpha+c\beta=0\,,
\\
a(\alpha-\gamma)+b\beta=0\,,
\\
b \alpha +d\beta=ad -bc\,,
\\
c(\alpha-\gamma)+d\beta=ad-bc\,.
\end{array}\right.
\end{equation}
The conditions in \eqref{20240618:eq2} are satisfied if and only if the matrix $\rmat$ defines a Lie algebra homomorphism
from $\mf g^*$ (with Lie bracket \eqref{eq:20240621-gstar}) to $\mf g$, and $\rmat^*$ defines a Lie algebra
homomorphism from $\mf g^*$ (with Lie bracket $[e_1^*,e_2^*]\, \tilde{} =(\alpha-\gamma) e_1^*+\beta e_2^*$) to $\mf g$
in agreement with the results from \cite{Dub89} mentioned in Remark \ref{rem:DP}.

For the simplest choice $\alpha=0$, $\beta=1$ in \eqref{eq:20240621-gstar}, and $\kmat=0$,  conditions \eqref{20240618:eq2} are satisfied for  $\rmat=\id_2$ (the identity matrix). In this case the Hamiltonian operator $K=A\mc S+B=\mc S^{-1}\circ A^t$ becomes
$$
K=
\begin{pmatrix}
u(\mc S-\mc S^{-1})\circ u& u(\mc S-1)\circ v\\v(1-\mc S^{-1})\circ u&u\mc S-\mc S^{-1}\circ u
\end{pmatrix}
\,,
$$
which is the well-known second Hamiltonian structure of the Toda lattice (see for example \cite{kmw13}).

\subsubsection{The Hamiltonian difference operator corresponding to the \texorpdfstring{$2$}{2}-dimensional Abelian lie group}\label{sec:ab}
Let
$$
G=\left\{\begin{pmatrix}u&0\\0&v\end{pmatrix}\Big| u,v\in\mb R, uv\neq0\right\}\subset GL_2(\mb R)
$$
be the 2-dimensional Abelian Lie group and let $\mf g=\mb R e_1\oplus\mb R e_2$ be its Lie algebra (which is Abelian as well). In this case, a basis of the space of left-invariant (as well as right-invariant) vector fields is the basis of coordinate vector fields $\frac{\partial}{\partial u}$ and $\frac{\partial}{\partial v}$. Then the matrix $A$ from Remark
\ref{rem:DP} is simply 
\begin{equation}\label{eq:Aabel}
A=\rmat=\begin{pmatrix} a & b \\c & d\end{pmatrix},\qquad\qquad\det A\neq 0.
\end{equation}
Following the procedure in \cite{Dub89}, we obtain that the bivector $\pi_G$ corresponding to the Lie algebra $\mf g^*$ with Lie bracket
\eqref{eq:20240621-gstar} is
$$
\pi_G=
\begin{pmatrix}
0& \alpha u+\beta v
\\
-\alpha u-\beta v & 0
\end{pmatrix}
\,.
$$
Let $\kmat=\begin{pmatrix}0&\gamma\\-\gamma&0\end{pmatrix}\in so_2(\mb R)$
Then, the matrix $B=B(u,v)$ in Remark \ref{rem:DP}
takes the form
\begin{align}\notag
B&=
\begin{pmatrix}
0 & \alpha u+\beta v+\gamma
\\
-\alpha u -\beta v-\gamma&0
\end{pmatrix}
\,.
\end{align}
According to the results in \cite{Dub89}, $K=A\mc S+B-\mc S^{-1}\circ A^T$ is a difference Hamiltonian operator if and only if the matrices $\rmat$ 
and $\rmat^*$ define a Lie algebra homomorphism from $\mf g^*$ with Lie bracket \eqref{eq:20240621-gstar} to the Abelian Lie algebra $\mf g$. Since both $\rmat$ and $\rmat^*$ are invertible, this condition forces $\alpha=\beta=0$. Hence,
$$
B=\begin{pmatrix}0 & \gamma \\ -\gamma & 0\end{pmatrix}\,,
\quad\gamma\in\mb R\,.
$$
On the other hand, let $A$ be as in \eqref{eq:Aabel} and let
$$
B=\begin{pmatrix}
0& f(u,v)\\
-f(u,v) & 0
\end{pmatrix}
\,.
$$
From \eqref{eq:cond1}-\eqref{eq:cond4} we have that $K=A\mc S+B-\mc S^{-1}\circ A^T$ is a difference Hamiltonian operator if and only if
$$
A\begin{pmatrix}
\frac{\partial f}{\partial u}\\
\frac{\partial f}{\partial v}
\end{pmatrix}
=A^T
\begin{pmatrix}
\frac{\partial f}{\partial u}\\
\frac{\partial f}{\partial v}
\end{pmatrix}
=0\,.
$$
Since, by assumption, $A$ is nondegenerate, we then have $f=\gamma\in\mb R$ in agreement with the results of \cite{Dub89}.

\subsubsection{The case when the matrix \texorpdfstring{$A$}{A} is degenerate}\label{sec:deg}

We study the Hamiltonian difference operators $K$ of the form \eqref{eq:diffop-1st}
with a matrix  $A$ which is degenerate and satisfies \eqref{eq:cond0}, namely, $A=A(u,v,u_1,v_1)$, where we are denoting $u^1=u$ and $u^2=v$.
First, we prove the following technical result.
\begin{lemma}\label{lem:col}If $A$ is a degenerate matrix of the form
\begin{equation}\label{eq:Adeg1}
A=\begin{pmatrix} a_{11}(u,v,u_1,v_1) & \mu(u_1,v_1)a_{11}(u,v,u_1,v_1)\\a_{21}(u,v,u_1,v_1) & \mu(u_1,v_1)a_{21}(u,v,u_1,v_1)\end{pmatrix},
\end{equation}
then there exists a point transformation leading it to the form
\begin{equation}\label{eq:Adeg-red}
\widetilde{A}=\begin{pmatrix}\widetilde{a}_{11}(u,v,u_1,v_1)&0\\\widetilde{a}_{21}(u,v,u_1,v_1)& 0\end{pmatrix}.
\end{equation}
\end{lemma}
\begin{proof} By \eqref{eq:tildeK}, the transformed difference operator $\tilde{K}$ is of the same order, and in particular the matrix $A$ of the form \eqref{eq:Adeg1} is transformed into
\begin{equation}
\widetilde{A}=\begin{pmatrix}\widetilde{a}_{11} & \widetilde{a}_{12}\\\widetilde{a}_{21} & \widetilde{a}_{22}\end{pmatrix}
\end{equation}
with
\begin{subequations}
\begin{align}
\widetilde{a}_{12}&
=\left(\frac{\dev \widetilde{u}}{\dev u}a_{11}+\frac{\dev \widetilde{u}}{\dev v}a_{21}\right)\mc S
\left(\frac{\dev \widetilde{v}}{\dev u}+\mu(u,v)\frac{\dev \widetilde{v}}{\dev v}\right)\,,
\label{eq1}\\
\widetilde{a}_{22}&
=\left(\frac{\dev \widetilde{v}}{\dev u}a_{11}+\frac{\dev \widetilde{v}}{\dev v}a_{21}\right)\mc S
\left(\frac{\dev \widetilde{v}}{\dev u}+\mu(u,v)\frac{\dev \widetilde{v}}{\dev v}\right)
\label{eq2}\,.\end{align}
\end{subequations}
From \eqref{eq1} and \eqref{eq2} we can get $\widetilde{a}_{12}=\widetilde{a}_{22}=0$ if
\begin{equation}\label{eq:Adeg-ch}
\frac{\dev\tilde{v}}{\dev u}+\mu(u,v)\frac{\dev\tilde{v}}{\dev v}=0.
\end{equation}
Equation \eqref{eq:Adeg-ch} is a homogeneous linear PDE for the function $\widetilde{v}=\widetilde{v}(u,v)$ that can be solved, for example, with the method of characteristics. 
 \end{proof}

Observe that a generic degenerate matrix $A$, depending on the aforementioned variables $(u,v,u_1,v_1)$ can be of the form \eqref{eq:Adeg-red} -- possibly after the swapping of columns --, of its transposed, or of the form
\begin{equation}\label{eq:Adeg-gen}
A=\begin{pmatrix} a_{11}(u,v,u_1,v_1) & a_{12}(u,v,u_1,v_1)\\\mu(u,v,u_1,v_1)a_{11}(u,v,u_1,v_1) & \mu(u,v,u_1,v_1)a_{12}(u,v,u_1,v_1)\end{pmatrix}
\end{equation}
with $a_{11}, a_{12}\neq 0$. The first step in our classification of Hamiltonian operators with degenerate leading term is showing that, in some coordinate system, \eqref{eq:Adeg-red} is indeed the only possible form. First of all, note that we can switch from the matrix $A$ to $A^{T}$ by exchanging $(\cS,u^l_n)\leftrightarrow(\cS^{-1},u^l_{-n})$, as it can be deduced from the form of \eqref{eq:diffop-1st}. Then, observe that $\det A=0$ means $a_{11}a_{22}-a_{12}a_{21}=0$. If $a_{12}=0$, then $a_{11}$ or $a_{22}$ must be equal to 0 too, therefore obtaining the ``one-column'' (i.e. \eqref{eq:Adeg-red}) or ``one-row'' case (by transposition of \eqref{eq:Adeg-red}). Finally, if $a_{12}\neq 0$, we can solve for $a_{21}$ and obtain $a_{21}=a_{11}a_{22}/a_{12}=\mu a_{11}$; we can also rewrite $a_{22}=\mu\, a_{12}$ therefore obtaining \eqref{eq:Adeg-gen}. Again, note that if $a_{11}$ vanishes we are back to the ``one-column'' case.

\begin{lemma}\label{lem:colfactor} Let $(A,B)$ define a Hamiltonian difference operator $K$ of the form \eqref{eq:diffop-1st} where $A$ is of the form \eqref{eq:Adeg-gen}.  Then
\begin{equation}\label{0110:eq1}
a_{12}(u,v,u_1,v_1)=\lambda(u_1,v_1)a_{11}(u,v,u_1,v_1)
\,,
\end{equation}
for some function $\lambda(x,y)$. Namely,
$A$ can be expressed as
$$
A=\begin{pmatrix} a_{11} & \lambda(u_1,v_1) a_{11}\\ \mu\, a_{11} & \lambda(u_1,v_1) \mu a_{11}\end{pmatrix},
$$
where we omit to denote the dependency on $(u,v,u_1,v_1)$ of the functions $a_{11}$ and $\mu$.
\end{lemma}
\begin{proof}
For $i=j=k=1$ and $i=2,j=k=1$, condition \eqref{eq:cond1} yields 
\begin{gather}\label{eq:lem8-pf1}
\left(\frac{\dev a_{11}}{\dev u_1}+\frac{\dev a_{11}}{\dev v_1}\cS\mu\right)\cS a_{11}-a_{11}\cS\frac{\dev a_{11}}{\dev u}-a_{12}\cS\frac{\dev a_{11}}{\dev v}=0\\\label{eq:lem8-pf2}
\left(\frac{\dev a_{11}}{\dev u_1}+\frac{\dev a_{11}}{\dev v_1}\cS\mu\right)\cS a_{12}-a_{11}\cS\frac{\dev  a_{12}}{\dev u}-a_{12}\cS\frac{\dev a_{12}}{\dev v}=0.
\end{gather}
Since we can assume that both $a_{11}$ and $a_{12}$ are different from 0, we may divide \eqref{eq:lem8-pf1} by $\cS a_{11}$ and \eqref{eq:lem8-pf2} by $\cS a_{12}$ and observe that the first two terms in both the equations are the same. Equating the remaining ones we obtain
\begin{equation}\label{eq:lem8-pf3}
a_{11}\cS\left(\frac{\dev_u a_{11}}{a_{11}}-\frac{\dev_u a_{12}}{a_{12}}\right)=a_{12}\cS\left(\frac{\dev_v a_{12}}{a_{12}}-\frac{\dev_v a_{11}}{a_{11}}\right).
\end{equation}
Observe that, since we assume that $a_{11}$ and $a_{12}$ are not vanishing, if the parenthesis in the LHS of \eqref{eq:lem8-pf3} vanishes, the parenthesis in the RHS of \eqref{eq:lem8-pf3} vanishes as well. Hence, \eqref{0110:eq1} holds.
If one of the two brackets does not vanish, then the other does not vanish as well and we can rewrite \eqref{eq:lem8-pf3} as
$$\frac{a_{11}}{a_{12}}=\cS\left(\frac{\frac{\dev_v a_{12}}{a_{12}}-\frac{\dev_v a_{11}}{a_{11}}}{\frac{\dev_u a_{11}}{a_{11}}-\frac{\dev_u a_{12}}{a_{12}}}\right)
$$
whose LHS can only depend on $(u,v,u_1,v_1)$, while the RHS depends on $(u_1,v_1,u_2,v_2)$.  This would imply that the ratio between $a_{11}$ and $a_{12}$ must be a function of $(u_1,v_1)$ only as in \eqref{0110:eq1}.
\end{proof}
By Lemmas \ref{lem:col} and \ref{lem:colfactor}, we can then focus on solving the full set of conditions \eqref{eq:cond1}-\eqref{eq:cond4} for $A$ as in \eqref{eq:Adeg-red}. One further simplification of the problem can be derived from \eqref{eq:cond1} for $i=k=1,j=2$, which yields $\partial_u a_{21}=0$. This means that we now need to specialise the set of equations \eqref{eq:cond-comp} to the simpler case
\begin{align}\label{eq:Adeg-simple}
A&=\begin{pmatrix}a_{11}(u,v,u_1,v_1)&0\\a_{21}(v,u_1,v_1)&0\end{pmatrix}&B&=\begin{pmatrix}0&b(u,v)\\-b(u,v)&0\end{pmatrix}.
\end{align}

The direct solution of the sets of equation \eqref{eq:cond-comp} for this choice is now still quite a cumbersome endeavour, which splits in the two subcases for $b\neq 0$ and $b=0$. Our approach is the following: first we find all the possible Hamiltonian difference operator $K$ of the form \eqref{eq:diffop-1st}, depending on a number of arbitrary functions -- this is the content of Propositions~\ref{prop:deg-B} and \ref{prop:deg-noB}; then we look for point transformations which simplify their expressions, fixing as many as these functional parameters as possible -- the final result is given in Theorem~\ref{thm:normal}.
\begin{proposition}\label{prop:deg-B}
Let $A$ and $B$ be as in equation \eqref{eq:Adeg-simple} and let us assume that $b\neq0$. Then, the pair  $(A,B)$ defines a Hamiltonian difference operator of the form \eqref{eq:diffop-1st} if and only if it is of one of the following two forms:
\begin{enumerate}[(1)]
\item \begin{align}\label{eq:Adeg-withb}
A&=\begin{pmatrix}f(u,v)f(u_1,v_1)g(v)F(u,v)&0\\
f(u_1,v_1)g(v)&0\end{pmatrix},&B&=\begin{pmatrix} 0 & \kappa f(u,v) g(v) \\ -\kappa f(u,v)g(v)&0\end{pmatrix}
\end{align}
where $F(u,v)$ is a solution of
\begin{equation}\label{eq:Adeg-withb-defF}
f(u,v)^2\dev_u F(u,v)=\dev_v f(u,v),
\end{equation}
for arbitrary nonvanishing constant $\kappa$, a single-variable function $g(x)$ and a two-variable function $f(x,y)$.
\item \begin{align}\label{eq:Adeg-withb-special}
A&=\begin{pmatrix}f(u,v)f(u_1,v_1)g(v,v_1)&0\\
0&0\end{pmatrix},&B&=\begin{pmatrix} 0 &  f(u,v) k(v) \\ - f(u,v)k(v)&0\end{pmatrix}
\end{align}
for an arbitrary single-variable function $k(x)$ and two two-variable functions $f(x,y)$ and $g(x,y)$.
\end{enumerate}
\end{proposition}
\begin{proof}
The form of the admissible Hamiltonian difference operators is obtained by directly solving the set of equations \eqref{eq:cond-comp} with the aid of a computer algebra software as \textsl{Mathematica}. We consider independently the three cases $a_{11},a_{21}\neq 0$, $a_{11}=0$, and $a_{21}=0$.

(i) $a_{11},a_{21}\neq0$. We start from \eqref{eq:Adeg-simple} and write \eqref{eq:cond2}, yielding
\begin{align*}
a_{11}\dev_{u_1}\cS b&=\left(\cS b\right) \dev_{u_1}a_{11},&a_{21}\dev_{u_1}\cS b&=\left(\cS b\right)\dev_{u_1}a_{21}.
\end{align*}
The two equations are identically solved if $b=0$, which is the case studied in Proposition \ref{prop:deg-noB}. If $b$ does not vanish, then the ratio between $a_{11}$ and $\cS b$, as well as the ratio between $a_{21}$ and $\cS b$, does not depend on $u_1$. Keeping into account the different functional dependency of $b$, $a_{11}$ and $a_{21}$ we can replace $a_{11}=R_1(u,v,v_1)\cS b$ and $a_{21}=R_2(v,v_1)\cS b$.

After this substitution and assuming that none of the factors identically vanishes, we obtain the following two equations from \eqref{eq:cond1}:
\begin{align}\label{eq:prop9-pf1a}
\cS\left( b R_2\right)\dev_{v_1}R_1+R_1\cS\left(R_2\dev_{v}b+R_1\dev_{u} b- b\dev_{u} R_1\right)&=0,\\
\label{eq:prop9-pf1b}
\cS\left( b R_2\right)\dev_{v_1}R_2+R_2\cS\left(R_2\dev_{v}b+R_1\dev_{u} b- b\dev_{u} R_1\right)&=0.
\end{align}
Dividing \eqref{eq:prop9-pf1a} by $R_1$ and \eqref{eq:prop9-pf1b} by $R_2$ we immediately see that they imply $\dev_{v_1}R_1/R_1=\dev_{v_1}R_2/R_2$, namely that the ratio between $R_1$ and $R_2$ does not depend on $v_1$. We can then replace $R_1(u,v,v_1)$ with $K(u,v)R_2(v,v_1)$. The remaining condition is
\begin{equation*}
\left(\cS b\right)\dev_{v_1}R_2+R_2\cS\left(\dev_v b+K\dev_u b-b\dev_u K\right)=0,
\end{equation*}
or equivalently
\begin{equation}\label{eq:prop9-pf2}
\frac{\dev_{v_1}R_2}{R_2}=-\cS\left(\frac{\dev_vb+K\dev_u b-b\dev_uK}{b}\right).
\end{equation}
Similarly, from \eqref{eq:cond3} we obtain
\begin{equation}\label{eq:prop9-pf3}
\frac{\dev_v R_2}{R_2}=\frac{\dev_v b+K\dev_u b-b\dev_uK}{b}.
\end{equation}
Note that the LHS of \eqref{eq:prop9-pf2} (resp. \eqref{eq:prop9-pf3}) can depend at most on $(v,v_1)$, while the RHS may depend on $(u_1,v_1)$ (resp. $(u,v)$). This means that $\dev_{v_1}R_2/R_2$ does not depend on $v$ (resp. $\dev_v R_2/R_2$ does not depend on $v_1$), which tells us that $R_2(v,v_1)=g(v)h(v_1)$. With this factorization, the LHS of \eqref{eq:prop9-pf2} reads $h'(v_1)/h(v_1)$ and the one of \eqref{eq:prop9-pf2} is $g'(v)/g(v)$. Taking into account the shift on the RHS of \eqref{eq:prop9-pf2} and comparing the RHS of the two conditions, we have $g'(v)/g(v)=-h'(v)/h(v)$, namely $h(v)=\rho/g(v)$ for a nonzero constant $\rho$. Summarizing, we have now
\begin{align}
a_{11}(u,v,u_1,v_1)&=\frac{\rho\,g(v)}{g(v_1)}K(u,v)b(u_1,v_1)\\ a_{21}(v,u_1,v_1)&=\frac{\rho\, g(v)}{g(v_1)}b(u_1,v_1)\,.
\end{align}
To simplify the expression, and noting that we have $b$ here an arbitrary function of the two variables $(u_1,v_1)$, we can redefine $\rho b(x,y)/g(y)=:f(x,y)$ from which, after calling $1/\rho=\kappa$, we have
\begin{align*}
a_{11}&=f(u_1,v_1)K(u,v)g(v),&a_{21}&=f(u_1,v_1)g(v),&b&=\kappa f(u,v)g(v).
\end{align*}
Rewriting \eqref{eq:prop9-pf3} in this form we are left with the equation
\begin{equation}\label{eq:prop9-pf4} 
f\dev_u K-\dev_u f K=\dev_v f,
\end{equation}
which we can regard as an ODE for $K$. Equation \eqref{eq:prop9-pf4} takes the simpler form \eqref{eq:Adeg-withb-defF} if we factorize $K=f(u,v)F(u,v)$; being a solution of an ODE in the variable $u$, the solution $F(u,v)$ additionally depends on an arbitrary function $k(v)$. The replacement of $K(u,v)$ in the expression for $a_{11}$ gives us \eqref{eq:Adeg-withb}.

(ii) $a_{11}=0$. We start once again from \eqref{eq:Adeg-simple} with $a_{11}=0$. From \eqref{eq:cond1} we obtain $\dev_{v_1}a_{21}=0$. Moreover, \eqref{eq:cond2} and \eqref{eq:cond3} yield the conditions
\begin{align*}
\cS\frac{\dev_u b}{b}&=\frac{\dev_{u_1}a_{21}}{a_{21}},&\frac{\dev_v b}{b}&=\frac{\dev_v a_{21}}{a_{21}}.
\end{align*} 
The different dependency of variables of the LHS and RHS of these conditions, as well as their form, give the solutions $a_{21}=f(u_1)g(v)$, $b=\kappa f(u)g(v)$. Note that this solution is of the same form as \eqref{eq:Adeg-withb} for $f(x,y)=f(x)$ and $k(v)=0$, from which $F=0$.

(iii) $a_{21}=0$. In this case, \eqref{eq:cond1} yields
\begin{equation}\label{eq:prop9-pf5}
\cS\frac{\dev_u a_{11}}{a_{11}}=\frac{\dev_{u_1}a_{11}}{a_{11}}.
\end{equation} 
Comparing the dependency of the terms on the LHS and on the RHS, we obtain the three conditions $\dev_{uu_1}\log a_{11}=0$, $\dev_{uv_1}\log a_{11}=0$, $\dev_{vu_1}\log a_{11}=0$, that are solved by $a_{11}=f(u,v)g(v,v_1)h(u_1,v_1)$. Finally, \eqref{eq:prop9-pf5} with such an ansatz gives $h(u_1,v_1)=f(u_1,v_1)h(v_1)$ -- the arbitrary single-function variable $h(v_1)$ can be absorbed into the arbitrary function $g(v,v_1)$ to give the form of $A$ in \eqref{eq:Adeg-withb-special}. Finally, for $b\neq 0$ both \eqref{eq:cond2} and \eqref{eq:cond3} yield the condition $\dev_u b/b=\dev_u f/f$, whose solution is $b(u,v)=f(u,v)k(v)$.
\end{proof}
\begin{proposition}\label{prop:deg-noB}
Let $A$ and $B$ be as in equation \eqref{eq:Adeg-simple} and let us assume that $b=0$.
Then $A$ defines a Hamiltonian difference operator of the form \eqref{eq:diffop-1st} if and only if is of one of the following three forms:
\begin{enumerate}[(1)]
\item
\begin{equation}
A=\begin{pmatrix}f(u)g(v)f(u_1)g(v_1)F(u,v)&0\\f(u_1)g(v_1)&0\end{pmatrix}
\end{equation}
where $F(u,v)$ is a solution of
\begin{equation}\label{eq:prop10-eq2}
f(u)g(v)^2\dev_u F(u,v)-g'(v)=0,
\end{equation}
for arbitrary single-variable functions $f(x), g(x)$.
\item
\begin{equation}\label{eq:prop:deg-noB2}
A=\begin{pmatrix}0&0\\f(u_1,v)&0\end{pmatrix}
\end{equation}
for an arbitrary two-variable function $f(x,y)$;
\item
\begin{equation}
A=\begin{pmatrix}f(u,v)f(u_1,v_1)g(v,v_1)&0\\0&0\end{pmatrix}
\end{equation}
for arbitrary two-variable functions $f(x,y), g(x,y)$.
\end{enumerate}
\end{proposition}
\begin{proof}
As in the previous case, we start from \eqref{eq:Adeg-simple} and consider separately the cases $a_{11},a_{21}\neq 0$, $a_{11}=0$, and $a_{21}=0$.

(i) $a_{11},a_{21}\neq 0$. From \eqref{eq:cond3} for $b=0$ and $i=j=k=2$ we obtain
$$
a_{11}\frac{\dev a_{21}}{\dev u_1}=a_{21}\frac{\dev a_{11}}{\dev u_1}
$$
which, where $a_{11}$ and $a_{21}$ are not zero, implies that $a_{11}=a_{21}R(u,v,v_1)$. After this substitution, the system \eqref{eq:cond1} is reduced to
\begin{gather}\label{eq:prop10-pf1}
\dev_{v_1}a_{21}+\left(\cS R\right)\dev_{u_1}a_{21}-a_{21}\dev_{u_1}\cS R=0\\\label{eq:prop10-pf2}
a_{21}\dev_{v_1}R+R\dev_{v_1}a_{21}+R\left(\cS R\right)\dev_{u_1}a_{21}-R a_{21}\dev_{u_1}\cS R=0.
\end{gather}
Observe that \eqref{eq:prop10-pf2} is $a_{21}\dev_{v_1}R+R\text{\eqref{eq:prop10-pf1}}$, so we conclude that $R$ does not depend on $v_1$. Finally, we are left with \eqref{eq:prop10-pf1} to solve. The equation can be regarded as a linear nonhomogeneous ODE for $\cS R$, that can then be expressed in terms of $a_{21}$, its derivatives, and an arbitrary function $h$ of single-variable. If we also require that $\cS R$ should not depend on $v$ because it is naturally a function of only $(u_1,v_1)$, this forces $a_{21}(v,u_1,v_1)$ to be of the form $f(u_1)g(v_1)$. Rewriting $\cS R(u_1,v_1)$ as $f(u_1)g(v_1)\cS F(u_1,v_1)$ we then reduce \eqref{eq:prop10-pf1} to the form \eqref{eq:prop10-eq2}.

(ii) $a_{11}=0$. The only condition we obtain from \eqref{eq:cond1} is $\dev_{v_1}a_{21}=0$; \eqref{eq:cond2} and \eqref{eq:cond3} are all identically satisfied by $a_{21}=f(u_1,v)$.

(iii) $a_{21}=0$. As for the  similar case in the proof of Proposition \ref{prop:deg-B}, \eqref{eq:cond1} admits as solution $a_{11}=f(u,v)f(u_1,v_1)g(v,v_1)$. For $b=0$ there are not further conditions.
\end{proof}
\begin{theorem}\label{thm:normal}
Let $K$ be a Hamiltonian difference operator of the form \eqref{eq:diffop-1st} with a matrix $A$ which is degenerate and satisfies \eqref{eq:cond0}. Then, up to point transformations, $K$ is of one of the following forms:
\begin{enumerate}
    \item Constant form, with
    \begin{equation}\label{eq:normaform-1}
    A=\begin{pmatrix}0 & 0\\1 &0\end{pmatrix},\qquad B=\begin{pmatrix}0 & \kappa \\-\kappa & 0\end{pmatrix},
    \end{equation}
    for an arbitrary constant $\kappa$.
    \item Type I form, with
    \begin{equation}\label{eq:normalform-2}
    A=\begin{pmatrix}g(v,v_1) & 0\\0 & 0\end{pmatrix},\qquad B=\begin{pmatrix} 0 & \sigma\\-\sigma & 0
    \end{pmatrix}
    \end{equation}
    for an arbitrary function $g(x,y)$ of two variables and $\sigma=0,1$.
    \item Type II form, with $A$ as in \eqref{eq:prop:deg-noB2} and $B=0$.
\end{enumerate}
\end{theorem}
\begin{proof}
Let $J=\left(\begin{smallmatrix}a&b\\c&d\end{smallmatrix}\right)$, with $\det J\neq 0$, $a=a(u,v)$, $b=b(u,v)$, $c=c(u,v)$, and $d=d(u,v)$ such that $\partial_v a=\partial_ub$, $\partial_vc=\partial_ud$, be the Jacobian matrix of the point transformation. From equation~\eqref{eq:tildeK}, the matrices $A$ and $B$ defining the Hamiltonian operator transform as
\begin{equation}
    \tilde{A}=JA\left(\cS J^T\right),\qquad\tilde{B}=JBJ^T.
\end{equation}

\emph{Case 1.} Let us consider $A$ and $B$ as in Proposition~\ref{prop:deg-B}(1). Under point transformations we have
\begin{align*}
\tilde{A}&=\begin{pmatrix}g(\cS a)(\cS f)\left(afF+b
\right)&g(\cS c)(\cS f)\left(afF+b
\right)\\ g(\cS a)(\cS f)\left(cfF+d
\right) &  g(\cS c)(\cS f)\left(cfF+d
\right)
\end{pmatrix},\\ \tilde{B}&=\begin{pmatrix}0 & \kappa(\det J) f g\\ -\kappa(\det J) f g&0\end{pmatrix}.
\end{align*}
For $c=0$ and $b=-afF$, the matrix $\tilde{A}$ assumes the form $\tilde{A}_{11}=\tilde{A}_{12}=\tilde{A}_{22}=0$, $\tilde{A}_{21}=g(\cS a)(\cS f)d$.  Observing that all the functions in $\tilde{A}_{12}$ depend only on the variables $(u,v)$, $\tilde{A}_{12}$ can assume a constant value for $af=\alpha$ constant and, independently, $gd=\beta$. Choosing the two constants in such a way that their product is $1$ gives $\tilde{A}$ of the form \eqref{eq:normaform-1}. Summarizing, our ansatz for the Jacobian matrix is
$$
J=\begin{pmatrix} \frac{\alpha}{f(u,v)} & -\alpha F(u,v)\\ 0 & \frac{1}{\alpha g(v)} \end{pmatrix}.
$$
It is immediate to see that $\det J=1/f(u,v)g(v)\neq 0$, so $\tilde{B}$ takes the form \eqref{eq:normaform-1} for any arbitrary nonzero $\kappa$ (as in \eqref{eq:Adeg-withb}); to be sure that $J$ is indeed a Jacobian matrix we need to verify $\dev_va=0=\dev_ub$ and $\dev_vc=\dev_ud$. This latter identity is indeed equivalent to the definition of $F(u,v)$ in \eqref{eq:Adeg-withb-defF}, so it is verified. We can repeat the same computation for $A$ as in Proposition~\ref{prop:deg-noB}(1). An ansatz for the Jacobian matrix taking $A$ to the constant form is
$$
J=\begin{pmatrix} \frac{\alpha}{f(u)g(v)} & -\alpha F(u,v)\\ 0 & \frac{1}{\alpha}\end{pmatrix}
$$
for an arbitrary constant $\alpha$; the final result is the same as \eqref{eq:normaform-1} with $\kappa=0$.

\emph{Case 2.} Take $A$ as in Proposition~\ref{prop:deg-B}(2). Under point transformations we have
$$
\tilde{A}=\begin{pmatrix}afg(\cS a)(\cS f)&afg(\cS c)(\cS f)\\ cfg(\cS a)(\cS f)&  cfg(\cS c)(\cS f)
\end{pmatrix},\qquad \tilde{B}=\begin{pmatrix}0 & (\det J) f k\\ (\det J) f k&0\end{pmatrix}.
$$
For $c=0$, $a=1/f(u,v)$ the matrix $\tilde{A}$ assumes the form $\tilde{A}_{12}=\tilde{A}_{21}=\tilde{A}_{22}=0$ and $\tilde{A}_{11}=g(v,v_1)$; $\tilde{B}$ assumes the form $\tilde{B}_{12}=-\tilde{B}_{21}= d\,k(v)$. Therefore, it is possible to obtain the form \eqref{eq:normalform-2} with $\sigma=1$ with a point transformation whose Jacobian matrix is
$$
J=\begin{pmatrix}\frac{1}{f(u,v)}&b(u,v)\\0 &\frac{1}{k(v)}\end{pmatrix},
$$
where $b(u,v)$ is a solution of $f^2\partial_ub+\dev_v f=0$. Note that also in the new coordinates we have $g=g(\tilde{v},\tilde{v}_1)$. On the other hand, an operator of the form of Proposition~\ref{prop:deg-noB}(3)is defined by a matrix $A$ of the same form. We can then use a point transformation with Jacobian matrix
$$
J=\begin{pmatrix}\frac{1}{f(u,v)}&b(u,v)\\0 &1\end{pmatrix}
$$ to obtain the same form for $\tilde{A}$ and $\tilde{B}=B=0$.

\emph{Case 3.} A matrix $A$ as in Proposition~\ref{prop:deg-noB}(2), under point transformations, becomes
$$
\tilde{A}=\begin{pmatrix}(\cS a)b f&(\cS c)b f\\(\cS a)df & (\cS c)d f\end{pmatrix}.$$
For a generic function $f(u_1,v)$ it is not possible to find a Jacobian matrix producing a constant matrix, because that would require a factorized form for the function $f=g(u_1)h(v)$. If that were the case, a Jacobian of the form
$$
J=\begin{pmatrix}0 & \frac{1}{\alpha h(v)}\\\frac{\alpha}{g(u)}&0\end{pmatrix}$$
would take the operator to the form \eqref{eq:normaform-1} with $\kappa=0$. On the other hand, if $f$ is not factorized we cannot simplify further.
\end{proof}

\begin{remark}\label{remark:Toda}
Observe that the matrices $(A,B)$ for the first Hamiltonian structure of the Toda lattice (see for example the review paper \cite{kmw13})
\begin{equation}\label{eq:HToda-0}
H=\begin{pmatrix} 0 & u(\cS-1) \\ (1-\cS^{-1})\circ u & 0
\end{pmatrix}
\end{equation}
are, respectively
$$
A=\begin{pmatrix}0&u\\0&0\end{pmatrix},\qquad B=\begin{pmatrix}0&-u\\u&0\end{pmatrix}.
$$
$A$ is degenerate, so $H$ must be equivalent, under point transformations, to one of the three normal forms of Theorem~\ref{thm:normal}. It is indeed of the constant form \eqref{eq:normaform-1} for $\kappa=1$ under the point transformation $\tilde{u}=v$, $\tilde{v}=\log u$. In Section \ref{ssec:cohoH0} we study the Poisson cohomology of the (equivalent for the exchange of $u$ and $v$ variables) Hamiltonian structure of the form
\begin{equation}\label{eq:HToda}
H_0=\begin{pmatrix}0 & \cS-1\\1-\cS^{-1} & 0\end{pmatrix}.
\end{equation}
\end{remark}
\begin{remark}
As previously remarked, the form $A=A(u,v,u_1,v_1)$ and $B=B(u,v)$ is a necessary condition for a difference operator $K$ of the form
\eqref{eq:diffop-1st} to be Hamiltonian if $A$ is nondegenerate. When $A$ is degenerate, Propositions \ref{prop:deg-B} and \ref{prop:deg-noB} produce only a subfamily -- to which all the examples known in the literature belong -- of all the possible solutions of \eqref{20240617:eq1}. For example, operators of the form
\begin{align*}
A&=\begin{pmatrix}g(v_M,\ldots,v_N)&0\\0&0\end{pmatrix},&B&=\begin{pmatrix}0&0\\0&0\end{pmatrix},
\end{align*}
$M\leq N$, are solutions too.
\end{remark}
\section{The Poisson cohomology}\label{sec:coho}
\subsection{The Poisson cohomology of a difference Hamiltonian structure}
Let $P$ be a Poisson bivector.
It is well known that from the property $[P,P]=0$ and the graded Jacobi identity for the Schouten bracket \eqref{eq:sch-jac} it follows that the adjoint action of $P$, which we denote by $d_P=[P,-]$, is a coboundary operator (i.e. a differential) on the graded space of local poly-vectors which we have identified with $\hF$ in Section \ref{sec:2.3}.
Explicitly (cf. \eqref{eq:sch-def}) it is given by
$$
d_P(Q)=-\sum_{l=1}^\ell \int \left(
\frac{\delta P}{\delta \theta_l}\frac{\delta Q}{\delta u^l}+\frac{\delta P}{\delta u^l}\frac{\delta Q}{\delta \theta_l}
\right)\,,
\quad
Q\in\hF
\,.
$$
This turns $(\hF,d_P)$ into the \emph{Poisson-Lichnerowicz complex}
$$
0\to\hF^0=\F\xrightarrow{d_P}\hF^1\xrightarrow{d_P}\hF^2\to\cdots
$$
Note that, in contrast with the finite dimensional case, there exist nontrivial $p$-vectors for any value of $p\geq0$, so the complex does not terminate.
The cohomology of this complex is the so-called \emph{Poisson cohomology} (originally introduced by Lichnerowicz \cite{l77} in finite dimensional Poisson geometry, see \cite{CW19} for the difference case); it controls the deformation theory of the Hamiltonian operator and carries information about the existence of bi-Hamiltonian hierarchies \cite{G01,kra,dSK}.

The Poisson cohomology for the normal form of a scalar Hamiltonian difference operator of $(-1,1)$-order has been previously computed in \cite{CW19}. In this section, we extend the results to some class of multi-component operators. We study in particular the operator $H_0$ in \eqref{eq:HToda}, which is the normal form of the $(-1,1)$-order Hamiltonian structure of the Toda lattice. Its importance in the study and classification of two-component difference systems will be discussed in Section \ref{sec:eg}, since it appears in several bi-Hamiltonian pairs.

Throughout this section we assume that $\mc A$ is a normal algebra of difference functions (see Section \ref{sec:2.1}) and that $\mc C=\bar{\mc C}$ is a field.

\subsection{The Poisson cohomology of an ultralocal structure}\label{ssec:cohoUL}
Hamiltonian structures of order $(0,0)$ are called \emph{ultralocal}. Because of the skewsymmetry requirement, there are not nontrivial ultralocal structures in the scalar case. The picture is richer in the multi-component case. Specializing the computation of the Jacobi identity \eqref{20240617:eq1} for a $\lambda$-bracket of the form $\{u^i_{\lambda}u^j\}=B^{ji}\in\mc A$ gives 
that $B$ needs to satisfy \eqref{eq:cond4}
and the conditions
\begin{equation}\label{eq:ultraloc-cond}
\sum_{s=1}^\ell \frac{\dev B^{ij}}{\dev u^s_n}\cS^n B^{sk} =0 \quad \text{for }i,j,k=1,\dots\ell\text{ and }n\neq0
\,.
\end{equation}

Assuming that $B$ is nondegenerate, then condition \eqref{eq:ultraloc-cond} implies $B^{ij}\in\mathcal A_0\subset \mathcal A$. Hence, nondegenerate ultralocal difference Hamiltonian structures are in one-to-one correspondence with nondegenerate Poisson bivectors on a $\ell$-dimensional manifold whose space of functions is $\mc A_0$ (see Remark \ref{rmk:A0}). Indeed, the density of the associate bivector
\begin{equation}\label{eq:bivUL}
P_B=\frac12\sum_{i,j=1}^\ell\int\theta_i B^{ij}\theta_j
\end{equation}
defines a Poisson bivector on $\A_0$.

The Poisson cohomology in the finite dimensional case, for $B$ nondegenerate, has been computed by Lichnerowicz himself \cite{l77}, who proved that it is isomorphic to the De Rham cohomology of the manifold. The same result holds true even if we consider ultralocal operators on the full space $\hF$. We have then
\begin{theorem}
Let $B$ be a nondegenerate ultralocal Hamiltonian operator with associated bivector $P_B$ given by \eqref{eq:bivUL}. Then $H^p(\hF,d_{P_B})=\delta_{p,0}\mc C$, $p\geq0$.
\end{theorem}
\begin{proof}

As previously observed, a nondegenerate ultralocal Hamiltonian operator is a nondegenerate Poisson structure on the underlying space of functions $\A_0$. As it is well known, such a structure is the inverse of a symplectic form, and thus there exists a reference frame on $\A_0$ (the Darboux coordinates) in which it is constant, namely we may assume $B^{ij}\in\mc C$.
Let $D_{P_B}:\hA\to \hA$ be the linear operator
\begin{equation}\label{0110:eq2}
D_{P_B}:=\sum_{i=1}^\ell\sum_{m\in\mb Z}\cS^m\left(\frac{\delta P_B}{\delta \theta_i}\right)\frac{\dev}{\dev u^i_m}
=\sum_{i,j=1}^\ell\sum_{m\in\mb Z}B^{ij}\theta_{j,m}\frac{\dev}{\dev u^i_m}
\,.
\end{equation}
Note that $\int D_{P_B}(Q)=-d_{P_B}Q$ and that, as seen in the proof of Theorem 2 of \cite{CW19}, $D_{P_B}^2=0$. Hence we have a complex $(\hA, D_{P_B})$.
We can regard this complex as the basic de Rham complex $(\widetilde{\Omega}(\mc A),\delta)$ defined in \cite[Section 5]{DSKVW18-1},
by identifying $\tilde\theta^i_m:=\sum_{j=1}^\ell B^{ij}\theta_{j,m}\in\hA$ with the differential $\delta u_m^i\in\widetilde{\Omega}(\mc A)$.
Under this identification, $(\hF, d_{P_B})$ is identified with the reduced de Rham complex $(\Omega(\mc A),\delta)$.
From Theorem 5.1 and Theorem 5.2 in \cite{DSKVW18-1} (see Remark \ref{rem:filtration}) we have that
$H^p(\hA,D_{P_B})=H^p(\hF,d_{P_B})=\delta_{p,0}\mc C$, $p\geq0$.
\end{proof}
\subsection{The Poisson cohomology of \texorpdfstring{$H_0$}{H0}}\label{ssec:cohoH0}
Let us consider the first Hamiltonian structure of Toda lattice given in \eqref{eq:HToda}. Its associated Poisson bivector \eqref{eq:bic-def}
is
\begin{equation}\label{eq:PToda}
P_0=\int\theta\left(\zeta_1-\zeta\right),
\end{equation}
where for simplicity we denote $\theta_{1,m}=\theta_m$ and $\theta_{2,m}=\zeta_m$, $m\in\mb Z$,
conjugate variables to $u_m^1=u_m$ and $u_m^2=v_m$.
In this section we compute its Poisson cohomology, showing that it is nontrivial but it is concentrated in the ultralocal part of $\hF$. In the sequel, given a closed $\omega\in\hF^p$, $p\geq0$, we denote $[\omega]=\omega+d_{P_0}(\hF^{p-1})\in H^p(\hF,d_{P_0})$.
\begin{theorem}\label{thm:cohoToda}
The Poisson cohomology $H^p(\hF,d_{P_0})$ is trivial for $p>2$. For $p\leq 2$ we have
\begin{align}
H^0(\hF,d_{P_0})
&=
\mc C[\smallint 1]\oplus\mc C[\smallint u]\oplus\mc C[\smallint v]\label{coh1}
\,,
\\
H^1(\hF,d_{P_0})&=\mc C[\smallint \theta]\oplus\mc C[\smallint \zeta]\oplus\mc C[\smallint\left(u\theta-v\zeta\right)]\,,\label{coh2}\\
H^2(\hF,d_{P_0})&=\mc C[\smallint\theta\zeta]
\label{coh3}\,.
\end{align}
\end{theorem}
The proof of Theorem \ref{thm:cohoToda} is obtained in two main steps.
Recall from \cite{CW19} that we have the differential on $\hA$
\begin{equation}\label{eq:DP0}
D_{P_0}=
\sum_{m\in\mb Z}
\left((\zeta_{m+1}-\zeta_{m})
\frac{\partial }{\partial u_m}
+(\theta_{m}-\theta_{m-1})
\frac{\partial }{\partial v_m}
\right)\,,
\end{equation}
such that $\int D_{P_0}(Q)=-d_{P_0}Q$, for every $Q\in\hF$.
As a first step we compute, in Proposition \ref{prop:coh1}, the cohomology of 
the auxiliary complex $(\hA,D_{P_0})$.
In the second step, since $\cS-1$ and $D_{P_0}$
commute we have that $((\cS-1)\hA,D_{P_0})\subset(\hA,D_{P_0})$ is a subcomplex.
Hence, we use the
short exact sequence of complexes
\begin{equation}\label{eq:short}
0\longrightarrow ( (\mc S-1)\hA,D_{P_0})
\stackrel{\alpha}{\longrightarrow}(\hA,D_{P_0})\stackrel{\beta}{\longrightarrow}
(\hF,d_{P_0})\longrightarrow 0
\,,
\end{equation}
where $\alpha$ is the inclusion map and $\beta=\int$ is the projection map,
together with the standard long exact sequence machinery to compute the cohomology of the complex $(\hF,d_{P_0})$.

Let us start by recalling the filtration $\hA^p_{n,i}$ of $\hA$ introduced in Section \ref{sec:2.3}. We define local homotopy operators $h_{n,i}:\hA_{n,i}^p\to\hA_{n,i}^{p-1}$, $n\in\mb Z$, $i=1,2$, as follows
\begin{equation}\label{eq:hom-op}
h_{n,1}=\left(\int \ud u_{n}\right)X_{n,1}
\quad\text{and}\quad
h_{n,2}=
\left(\int \ud v_{n}\right)X_{n,2}
\,,
\end{equation}
where the integral map \eqref{integral} is extended to an endomorphism of $\hA_{n,i}^p$ by letting it act trivially on the variables $\theta_m$ and $\zeta_m$, and $X_{n,1}$ and $X_{n,2}$ are the following vector fields
$$
X_{n,1}
=
\frac12\left(\sum_{m>n}\frac{\partial}{\partial \zeta_{m}}
-\sum_{m\leq n}\frac{\partial}{\partial \zeta_{m}}\right)
\quad\text{and}\quad
X_{n,2}
=\frac12\left(\sum_{m\geq n}\frac{\partial}{\partial \theta_{m}}
-
\sum_{m<n}\frac{\partial}{\partial \theta_{m}}
\right)\,.
$$
For example, for $\omega=f P\in\hA^p_{n,1}$, where $f\in\mc A_{n,1}$ and $P\in\mc C[\theta_m,\zeta_m\mid m\in\mb Z]$ of degree $p$, we have
\begin{equation}\label{20241106:eq4}
h_{n,1}(\omega)
=\frac12\left(\int f \,\ud u_{n}\right)
\left(
\sum_{m>n}\frac{\partial P}{\partial \zeta_{m}}
-\sum_{m\leq n}\frac{\partial P}{\partial \zeta_{m}}
\right)
\,.
\end{equation}
By a direct computation we have ($m,n\in\mb Z$)
\begin{equation}\label{20241106:eq2}
X_{n,1}(\zeta_{m+1}-\zeta_m)=\delta_{m,n}\,,
\qquad
X_{n,2}(\theta_{m}-\theta_{m-1})=\delta_{m,n}\,.
\end{equation}
\begin{lemma}\label{lem:coh1}
Let $\omega\in\hA^p_{n,i}$. The operators $h_{n,i}$ satisfy the following
homotopy condition
\begin{equation}\label{eq:homotopy}
h_{n,i}(D_{P_0}(\omega))+D_{P_0}(h_{n,i}(\omega))-\omega\in\hA^p_{n,i-1}
\,.
\end{equation}
\end{lemma}
\begin{proof}
It is clear from \eqref{eq:DP0} that $D_{P_0}(\hA^p_{n,i})\subset\hA_{n,i}^{p+1}$. Hence,
$h_{n,i}(D_{P_0}(\omega))+D_{P_0}(h_{n,i}(\omega))-\omega\in\hA_{n,i}^p$.
To prove the claim we need to show that
\begin{equation}\label{eq:toprove}
\frac{\partial }{\partial u_{n}^i}\left(h_{n,i}(D_{P_0}(\omega))+D_{P_0}(h_{n,i}(\omega))-\omega
\right)=0
\end{equation}
(where we recall that we have set $u_n^1=u_n$ and $u_n^2=v_n$).

Let us prove \eqref{eq:toprove} for $i=1$. By linearity,
it suffices to check it for $\omega=fP$, where $f\in\mc A_{n,1}$ and $P\in\mc C[\theta_m,\zeta_m\mid m\in\mb Z]$ of degree $p$.
From the definition of the differential $D_{P_0}$ in \eqref{eq:DP0}, the definition of the homotopy operator $h_{n,1}$ and the first identity in \eqref{20241106:eq2} we have
\begin{equation}\label{20241106:eq1}
\begin{split}
h_{n,1}(D_{P_0}(\omega))&=
\left(\int\frac{\partial f}{\partial u_n}\,\ud u_n \right)P
\\
&-\sum_{m\in\mb Z}\left(\int\frac{\partial f}{\partial u_m}\,\ud u_n \right) (\zeta_{m+1}-\zeta_m)X_{n,1}(P)
\\
&
-\sum_{m\in\mb Z}\left(\int\frac{\partial f}{\partial v_m}\,\ud u_n \right) (\theta_{m}-\theta_{m-1})X_{n,1}(P)
\,.\end{split}
\end{equation}
On the other hand, using \eqref{20241106:eq4} and the definition of $D_{P_0}$ in \eqref{eq:DP0} we get
\begin{equation}\label{20241106:eq3}
\begin{split}
D_{P_0}(h_{n,1}(\omega))&=
-\sum_{m\in\mb Z}\frac{\partial }{\partial u_m}\left(\int f\,\ud u_n \right) (\zeta_{m+1}-\zeta_m)X_{n,1}(P)
\\
&
-\sum_{m\in\mb Z}\frac{\partial }{\partial v_m}\left(\int f\,\ud u_n \right) (\theta_{m}-\theta_{m-1})X_{n,1}(P)
\,.
\end{split}
\end{equation}
Equation \eqref{eq:toprove} for $\omega=fP$ follows from \eqref{20241106:eq1} and \eqref{20241106:eq3} using the fact that partial derivatives commute and that
$\frac{\partial }{\partial u_n}(\int g \,\ud u_n)=g$, for every $g\in\mc A_{n,1}$. The proof of \eqref{eq:toprove} for $i=2$ is done similarly using the second identity in \eqref{20241106:eq2} to compute $h_{n,2}(D_{P_0}(\omega))$.
\end{proof}
Let us assume that $\omega\in\hA^p_{n,i}$ is such that $D_{P_0}(\omega)=0$. Then, by \eqref{eq:homotopy}
we have that $\omega=D_{P_0}(h_{n,i}(\omega))+\widetilde \omega$, for some $\widetilde\omega\in\hA^p_{n,i-1}$ such that
$D_{P_0}(\widetilde \omega)=0$. Repeating the same argument finitely many times we have that
$\omega=D_{P_0}(\bar{\omega})+\widetilde \omega$ for some $\bar{\omega}\in \hA^{p-1}_{n,i}$and $\widetilde\omega\in\hA^p_{0,0}$.
Hence,
\begin{equation}\label{20241101:eq3}
\ker\left( D_{P_0}:\hA^{p}\rightarrow\hA^{p+1}\right)
=D_{P_0}(\hA^{p-1})+\hA^{p}_{0,0}
\,,
\end{equation}
for every $p\geq0$.

Note that $\hA^p_{0,0}$ is spanned (over $\mc C$) by monomials (recall that we denoted $\theta_{1,n}=\theta_{n}$ and $\theta_{2,n}=\zeta_n$)
\begin{equation}\label{20241101:eq1}
\omega=\theta_{i_1,n_1}\theta_{i_2,n2}\dots\theta_{i_p,n_p}
\,,
\end{equation}
where $(\epsilon(n_1),i_1)>(\epsilon(n_2),i_2)>\dots(\epsilon(i_p),n_p)$ in the lexicographic order. We define the degree of $\omega$ as in \eqref{20241101:eq1} by letting
\begin{equation}\label{eq:degree}
\deg{\omega}=\sum_{k=1}^p\epsilon(n_p)\,.
\end{equation}
We then have the direct sum decomposition
$$
\hA^p_{0,0}=\bigoplus_{n\geq0}\hA^p_{0,0}[n]
\,,
$$
where $\hA^p_{0,0}[n]$ is the subspace consisting of homogeneous elements of degree \eqref{eq:degree} equal to $n$.
\begin{lemma}\label{lem:coh2}
Let $\omega\in\hA^p_{0,0}$. There exist $\widetilde\omega\in \hA^p_{0,0}[0]$ such that
\begin{equation}\label{20241101:eq2}
\omega-\widetilde{\omega}\in D_{P_{0}}(\hA^{p-1})
\,.
\end{equation}
\end{lemma}
\begin{proof}
By linearity it suffices to check \eqref{20241101:eq2} for $\omega$ as in \eqref{20241101:eq1} with $\deg{\omega}>0$ (and thus $\epsilon(n_1)>0$, namely $n_1\neq0$.). 
We have four possible cases.
\begin{enumerate}[(a)]
\item
If $i_1=1$ and $n_1>0$, then, using the definition of $D_{P_0}$ we have
$$
\omega=\theta_{n_1-1}\theta_{i_2,n_2}\dots\theta_{i_p,n_p}+D_{P_0}(v_{n_1}\theta_{i_2,n_2}\dots\theta_{i_p,n_p})
\,,
$$
and $\deg(\theta_{n_1-1}\theta_{i_2,n_2}\dots\theta_{i_p,n_p})<\deg(\omega)$ since $\epsilon(n_1-1)<\epsilon(n_1)$.
\item
If $i_1=1$ and $n_1<0$, we have
$$
\omega=\theta_{n_1+1}\theta_{i_2,n_2}\dots\theta_{i_p,n_p}-D_{P_0}(v_{n_1+1}\theta_{i_2,n_2}\dots\theta_{i_p,n_p})
\,,
$$
and $\deg(\theta_{n_1+1}\theta_{i_2,n_2}\dots\theta_{i_p,n_p})<\deg(\omega)$ since $\epsilon(n_1+1)<\epsilon(n_1)$.
\item
If $i_1=2$ and $n_1>0$,  we have
$$
\omega=\zeta_{n_1-1}\theta_{i_2,n_2}\dots\theta_{i_p,n_p}+D_{P_0}(u_{n_1-1}\theta_{i_2,n_2}\dots\theta_{i_p,n_p})
\,,
$$
and $\deg(\zeta_{n_1-1}\theta_{i_2,n_2}\dots\theta_{i_p,n_p})<\deg(\omega)$ since $\epsilon(n_1-1)<\epsilon(n_1)$.
\item
If $i_1=2$ and $n_1<0$, we have
$$
\omega=\zeta_{n_1+1}\theta_{i_2,n_2}\dots\theta_{i_p,n_p}-D_{P_0}(u_{n_1}\theta_{i_2,n_2}\dots\theta_{i_p,n_p})
\,,
$$
and $\deg(\zeta_{n_1+1}\theta_{i_2,n_2}\dots\theta_{i_p,n_p})<\deg(\omega)$ since $\epsilon(n_1+1)<\epsilon(n_1)$.
\end{enumerate}
We have thus found $\widetilde\omega_1\in \hA^p_{0,0}$ with $\deg(\widetilde{\omega}_1)<\deg(\omega)$ 
such that $\omega-\widetilde{\omega}_1\in D_{P_{0}}(\hA^{p-1})$. Iterating this process finitely many times we arrive at \eqref{20241101:eq2}.
\end{proof}
We are then able to describe the cohomology of the differential $(\hA,D_{P_0})$. For an element $\omega\in\hA^p$ such that
$D_{P_0}(\omega)=0$, we denote by $[\omega]=\omega+D_{P_0}(\hA^{p-1})\in H^{p}(\hA,D_{P_0})$ its cohomology class.
\begin{proposition}\label{prop:coh1}
The cohomology $H(\hA,D_{P_0})$ is finite-dimensional and we have
\begin{equation}
H(\hA,D_{P_0})\cong \mc C\oplus\mc C\theta\oplus\mc C\zeta\oplus\mc C\theta\zeta
\,.
\end{equation}
In particular, $H^0(\hA,D_{P_0})=\mc C[1]$, $H^1(\hA,D_{P_0})=\mc C[\theta]\oplus\mc C[\zeta]$,
$H^2(\hA,D_{P_0})=\mc C[\theta\zeta]$
and $H^p(\hA,D_{P_0})=0$, for $p\geq3$.
\end{proposition}
\begin{proof}
From \eqref{20241101:eq3} and \eqref{20241101:eq2} we then have
$$
\ker\left( D_{P_0}:\hA^{p}\rightarrow\hA^{p+1}\right)
=D_{P_0}(\hA^{p-1})+\hA^{p}_{0,0}[0]
\,,
$$
for every $p\geq0$. To conclude the proof we note that
$$
\hA^p_{0,0}[0]=
\left\{
\begin{array}{cc}
\mc C\,, & p=0\,,
\\
\mc C\theta\oplus\mc C\zeta\,,&p=1\,,
\\
\mc C\theta \zeta\,,& p=2\,,
\\
0\,,&p\geq3\,.
\end{array}
\right.
$$
and that $D_{P_0}(\hA^{p-1})\cap \hA^p_{0,0}[0]=0$ for $p\leq 2$.
\end{proof}
Next, we use Proposition \ref{prop:coh1} to conclude the proof of Theorem \ref{thm:cohoToda}.
The short exact sequence
\eqref{eq:short} induces the long exact sequence
\begin{equation}\label{eq:long}
\begin{split}
0&\longrightarrow H^0((\cS-1)\hA,D_{P_0})
\stackrel{\alpha_0}{\longrightarrow}H^0(\hA,D_{P_0})\stackrel{\beta_0}{\longrightarrow}
H^0(\hF,d_{P_0})\stackrel{\gamma_0}{\longrightarrow}
\\
&\stackrel{\gamma_0}{\longrightarrow} H^1((\cS-1)\hA,D_{P_0})
\stackrel{\alpha_1}{\longrightarrow}H^1(\hA,D_{P_0})\stackrel{\beta_1}{\longrightarrow}
H^1(\hF,d_{P_0})\stackrel{\gamma_1}{\longrightarrow}
\\
&\stackrel{\gamma_1}{\longrightarrow} H^2( (\cS-1)\hA,D_{P_0})
\stackrel{\alpha_2}{\longrightarrow}H^2(\hA,D_{P_0})\stackrel{\beta_2}{\longrightarrow}
H^2(\hF,d_{P_0})\stackrel{\gamma_2}{\longrightarrow}
\\
&\stackrel{\gamma_2}{\longrightarrow} H^3((\cS-1)\hA,D_{P_0})
\stackrel{\alpha_3}{\longrightarrow}H^3(\hA,D_{P_0})\stackrel{\beta_3}{\longrightarrow}
H^3(\hF,d_{P_0})\stackrel{\gamma_3}{\longrightarrow}
\\
&\stackrel{\gamma_3}{\longrightarrow} H^3((\cS-1)\hA,D_{P_0})
\dots
\,.
\end{split}
\end{equation}
We start by computing the cohomology of the complex $((\mc S-1)\hA,D_{P_0})$ and the maps $\alpha_p$ in \eqref{eq:long}.
\begin{lemma}\label{lem:coh}
\begin{enumerate}[(a)]
\item
We have that $H^0((\mc S-1)\hA,D_{P_0})=0$ and, for $p\geq1$, 
\begin{equation}\label{eq:subcomplex}
H^p((\mc S-1)\hA,D_{P_0})\cong H^p(\hA,D_{P_0})
\,.
\end{equation}
\item
All the linear maps $\alpha_p$, $p\geq0$, in the long exact sequence \eqref{eq:long} are trivial.
\end{enumerate}
\end{lemma}
\begin{proof}
Since $D_{P_0}$ and $(\cS-1)$ commute, we clearly have $\ker(D_{P_0}|_{(\cS-1)\hA^p})=(\mc S-1)\ker(D_{P_0}|_{\mc A^p})\subset\hA^p$,
for every $p\geq0$, and $D_{P_0}((\mc S-1)\hA^{p-1})=(\mc S-1)D_{P_0}(\mc A^{p-1})\subset\hA^p$,
for every $p\geq1$.
Note that $\mc S-1:\hA^p\rightarrow\hA^p$ is injective for $p\geq1$ while for $p=0$ we have $\hA^0=\mc A$  and $\ker (\mc S-1)|_{\mc A}=\mc C$.
Hence, we have
$$
H^0((\mc S-1)\hA,D_{P_0})
=\ker(D_{P_0}|_{(\mc S-1)\mc A})=(\mc S-1)\ker(D_{P_0}|_{\mc A})=(\mc S-1)\mc C=0\,,
$$
and, for $p\geq1$,
\begin{align*}
H^p((\mc S-1)\hA,D_{P_0})
&=\frac{\ker(D_{P_0}|_{(\mc S-1)\hA^p})}{D_{P_0}((\mc S-1)\hA^{p-1})}
=\frac{(\mc S-1)\ker(D_{P_0}|_{\hA^p})}{(\mc S-1)D_{P_0}(\hA^{p-1})}
\\
&\cong 
\frac{\ker(D_{P_0}|_{\hA^p})}{D_{P_0}(\hA^{p-1})}
=H^p(\hA,D_{P_0})\,,
\end{align*}
where in the third identity we used the fact that $\mc S-1$ is injective. This proves part (a). The claim of part (b) is clear for $p\neq 1,2$ using part (a) and Proposition \ref{prop:coh1}. From part (a) and Proposition \ref{prop:coh1} we also have that 
$H^1((\mc S-1)\hA,D_{P_0})$ is the linear (over $\mc C$) span of the cohomology classes of $(\mc S-1)\theta$ and $(\cS-1)\zeta$. Since 
$$
(\cS-1)\theta=D_{P_0}(v_1)
\quad\text{and}\quad
(\cS-1)\zeta=D_{P_0}(u)
\,,
$$
they are mapped by $\alpha_1$ to the trivial cohomology class in $H(\hA,D_{P_0})$.
Similarly, 
$H^2((\mc S-1)\hA,D_{P_0})$ is the linear (over $\mc C$) span of the cohomology class of $(\mc S-1)\theta\zeta$
which is mapped by $\alpha_2$ to the trivial cohomology class in in $H(\hA,D_{P_0})$ since
$$
(\cS-1)\theta\zeta
=
D_{P_0}(v_1\zeta_1-u\theta)
\,.
$$
This concludes the proof of part (b).
\end{proof}
Finally, we are able to prove Theorem \ref{thm:cohoToda}.
\begin{proof}[Proof of Theorem \ref{thm:cohoToda}]
By Proposition \ref{prop:coh1} and Lemma \ref{lem:coh}(a), using the long exact sequence \eqref{eq:long} we immediately have that $H^p(\hF,d_{P_0})$ is trivial for $p>2$.

By Lemma \ref{lem:coh} we have that $H^0((\cS-1)\hA,D_{P_0})$ is trivial. Moreover, by Lemma \ref{lem:coh}(b), $\alpha_1$ is trivial, hence $\gamma_0$ is surjective. Then we have the following short exact sequence from \eqref{eq:long}:
\begin{equation}\label{eq:long1}
\begin{split}
0&\longrightarrow 
H^0(\hA,D_{P_0})\stackrel{\beta_0}{\longrightarrow}
H^0(\hF,d_{P_0})\stackrel{\gamma_0}{\longrightarrow}H^1((\mc S-1)\hA,D_{P_0})
\longrightarrow0\,,
\end{split}
\end{equation}
where the map $\beta_0$ sends $[1]\in H^0(\hA,D_{P_0})$ to $[\int 1]\in H^0(\hF,d_{P_0})$.
Using Proposition \ref{prop:coh1} and the isomorphism \eqref{eq:subcomplex}, from \eqref{eq:long1} we get that
$$
H^0(\hF,d_{P_0})=\mc C[\smallint 1]\oplus\mc C[\omega_1]\oplus\mc C[\omega_2]
\,,
$$
where $[\omega_i]\in H^0(\hF,d_{P_0})$, $i=1,2$, are such that $\gamma_0([\omega_1]))=[(\cS-1)\theta]$ and
$\gamma_0([\omega_2]))=[(\cS-1)\zeta]$.
Note that $d_{P_0}(\int u)=d_{P_0}(\int v)=0$, and, by an explicit computation using \eqref{eq:DP0} and the definition of
the connecting homomorphism $\gamma_0$ we have
$$
\gamma_0([\smallint u])=[D_{P_0}(u)]=[(\mc S-1)\zeta]
\,,
$$
and
$$
\gamma_0([\smallint v])=[D_{P_0}(v)]=[(\mc S-1)\theta_{-1}]=[(\mc S-1)\theta]
\,,
$$
where in the last identity we used the fact that $\theta-\theta_{-1}=D_{P_0}(v)$.
This proves \eqref{coh1}.
Similarly, since $\alpha_2$ is trivial, then $\gamma_1$ is surjective and, from \eqref{eq:long}, we get the short exact sequence
\begin{equation}\label{eq:long2}
\begin{split}
0&\longrightarrow 
H^1(\hA,D_{P_0})\stackrel{\beta_1}{\longrightarrow}
H^1(\hF,d_{P_0})\stackrel{\gamma_1}{\longrightarrow}H^2((\mc S-1)\hA,D_{P_0})
\longrightarrow0\,,
\end{split}
\end{equation}
where 
and $\beta_1[\theta]=[\int\theta]$, $\beta_1[\zeta]=[\int\zeta]$. From \eqref{eq:long2} it follows that
$$
H^1(\hF,d_{P_0})=\mc C[\smallint \theta]\oplus\mc C[\smallint\zeta]\oplus\mc C[\omega]
\,,
$$
where $[\omega]\in H^1(\hF,d_{P_0})$ is such that  $\gamma_1([\omega])=[\theta\zeta]$.
A choice for $\omega$ is $\omega=\int v\zeta-u\theta$.
Indeed, using the properties of the integral map and the definition of $d_{P_0}$ given after \eqref{eq:DP0} it is straightforward to check that $d_{P_0}(\omega)=0$, and that, by an explicit computation using the definition of the connecting homomorphism $\gamma_1$, we have 
$$
\gamma_1[\omega]=[D_{P_0}(v\zeta-u\theta)]
=[\theta\zeta_{1}-\theta_{-1}\zeta]=[(\mc S-1)\theta_{-1}\zeta]=[(\mc S-1)\theta\zeta]
\,,$$
where in the last identity we used the fact that $(\theta-\theta_{-1})\zeta=D_{P_0}(v\zeta)$.
This proves \eqref{coh2}.
Finally, since $\alpha_2$ is trivial and $H^3((\mc S-1)\hA,D_{P_0})=0$ we have the isomorphism
$$
0\longrightarrow H^2(\hA,D_{P_0})\stackrel{\beta_2}{\longrightarrow}
H^2(\hF,d_{P_0})\longrightarrow0
\,,
$$
where $\beta_2[\theta\zeta]=[\int \theta\zeta]$, thus proving \eqref{coh3} and concluding the proof.
\end{proof}
\begin{remark}
The representative $\frac12\int\theta\zeta$ of the cohomology class $[\frac12\int \theta\zeta]\in H^2(\hF,d_{P_0})$
corresponds via \eqref{eq:bic-def} to the ultralocal Hamiltonian operator
$$
H_{(ul)}=\begin{pmatrix} 0 & 1 \\-1 &0 \end{pmatrix}
\,.$$
\end{remark}
\begin{remark}
Observe that the normal forms of
Hamiltonian difference operators provided by Theorem \ref{thm:normal} are different from the
Dubrovin-Novikov normal forms of hydrodynamic differential Hamiltonian operators \cite{DN83}. Despite this fact, by Theorem \ref{thm:cohoToda},
we have an isomorphism of graded vector spaces between the Poisson cohomology $H(\hF,d_{P_0})$ and the Poisson cohomology of a 2-component non-degenerate Dubrovin-Novikov differential Hamiltonian operator \cite{G01}.
\end{remark}

\section{Local bi-Hamiltonian pairs in two components}\label{sec:eg}

In this section we present a few examples of compatible pairs of local Hamiltonian difference structures in two components, taken by \cite{kmw13}.

As it is well known, the Poisson bivectors $(P_1,P_2)$ associated to a pair of compatible Hamiltonian structures satisfy the identity $[P_1+\lambda P_2,P_1+\lambda P_2]=0$, for every  $\lambda\in\mc C$. In particular, the compatibility condition $[P_1,P_2]=0$ means that the bivector $P_2$ is a cocycle of $d_{P_1}$. 


In the previous section, we computed the Poisson cohomology of $P_0$ given in \eqref{eq:PToda}; Theorem \ref{thm:cohoToda} guarantees that all the bivectors $P$ compatible with $P_0$ are a linear combination of the only nontrivial cocycle $P_{(ul)}=\int\theta\zeta$ with trivial deformations, namely that there exist $\alpha\in\mc C$, $X\in\hF^1$ such that $P=\alpha P_{(ul)}+[P_0,X]$. Observe, in particular, that the vector $X$ must be nonzero if $P$ is not ultralocal.

Point transformations do not affect the Poisson cohomology; this means that, where there exists a point transformation taking a bivector $P_1$ to the form $P_0$, for any $P$ such that $[P_1,P]=0$ there must exist $X\in\hF^1$ such that $P=\alpha P'_{(ul)}+d_{P_1}X$, where we denoted as $P'_{(ul)}$ the inverse transformed form of the bivector $P_{(ul)}$.

When looking for the explicit expression of the local 1-vector $X=\int f(u,v,\ldots)\theta+g(u,v,\dots)\zeta$ in the following sections we need to make an ansatz on the variable dependency of its components $f,g\in\mc A$. To this end, we preliminary observe the expression for the adjoint action of a local bivector $P=\frac{1}{2}\int\sum_{i,j=1}^\ell\sum_{s\in\mb Z}\theta_i K^{ij}_s\theta_{j,s}$ on a local 1-vector $X=\int \sum_{i=1}^\ell X^i\theta_i$. A direct computation using \eqref{eq:sch-def} shows that
\begin{equation}\label{eq:adjaction}
[P,X]=\sum_{i,j,l=1}^\ell\sum_{m,s\in\mb Z}
\int\left(\theta_i\frac{\dev X^i}{\dev u^l_m}\left(\cS^m K^{lj}_s\right)\theta_{j,s+m}-\frac12\theta_i\left(\cS^mX^l\right)\frac{\dev K^{ij}_s}{\dev u^l_m}\theta_{j,s}\right).
\end{equation}
Recall now the filtered spaces $\A_{n,i}$ defined in Section~\ref{sec:2.1}
and the filtration given by equation \eqref{eq:filtration} and let $(-N,N)$ be the order of the matrix difference operator corresponding to $P$. Assume that $X^i\in\A_{p+1,0}$; from the first term of \eqref{eq:adjaction} we can conclude that the order of the difference operator associated to $[P,X]$ is $(-n,n)$, with $n\leq N+|p|$. This imposes a lower bound on $p$; in general there is no upper bound on the number of shifted variables on which the characteristics $X^i$ of $X$ should depend -- we start looking from the minimal possible dependency, namely $X^i\in\A_{N-n+1,0}$, and increase the filtration in case we cannot find a solution there.

We conclude the section showing how our cohomological approach to the Hamiltonian structures can be used to derive new examples of local bi-Hamiltonian pairs.

\subsection{Toda lattice}\label{sec:toda}
The first Hamiltonian structure of the Toda lattice is \eqref{eq:HToda-0}, whereas the second one is
\begin{equation}\label{eq:HToda2-0}
H_2= \begin{pmatrix} u(\cS-\cS^{-1})\circ u & u(\cS-1)\circ v\\ v(1-\cS^{-1})\circ u & u\cS-\cS^{-1}\circ u \end{pmatrix}
\,.
\end{equation}
As already claimed, the change of coordinate $u'=\log u$ and $v'=v$ takes the first Hamiltonian structure $H$ given in \eqref{eq:HToda-0} to the constant form $H_0$ of \eqref{eq:HToda}. The same change of coordinates applied to $H_2$ gives
\begin{equation}\label{eq:HToda2}
\tilde{H}_2=\begin{pmatrix}\cS-\cS^{-1} & (\cS-1)\circ v\\v(1-\cS^{-1}) & e^u\cS-\cS^{-1}\circ e^u \end{pmatrix}
\,.
\end{equation}
Clearly, the change of coordinates does not affect the compatibility of the two structures. Let us denote by $P_2$ the Poisson bivector associated to $H_2$ in \eqref{eq:HToda2-0}, by $\tilde{P}_2$ the one associated to $\tilde{H}_2$, by $P_1$ the one associated to $H$ of \eqref{eq:HToda-0} and, finally, $P_0$ is given by \eqref{eq:PToda} as above mentioned. From $[P_1,P_2]=0$ we have that $\tilde{P_2}$ is a cocycle in $H^2(d_{P_0})$. Because the only nontrivial cocycle in $H^2(d_{P_0})$ is the constant ultralocal operator $H_{(ul)}$, there must exist an element $\tilde{X}\in\hF^1$ such that \emph{at least} the coefficients of $\theta\theta_1$, $\theta\zeta_1$, $\zeta\zeta_1$ of $\tilde{H}_2$ are given by the corresponding coefficients of $d_{P_0}\tilde{X}$; on the other hand the coefficients of $\theta\zeta$ may differ only up to an additive constant. In the original coordinates of \eqref{eq:HToda-0}, we have $P_2=[P_1,X]+\alpha P_{(ul)}$ for some $X\in\hF^1$, and $P_{(ul)}=\frac12\int u \theta \zeta$. We can then compute $X$ directly in this coordinate system.

Explicitly, the two Poisson bivectors $P_1$ and $P_2$ are of the form
\begin{align}\label{eq:P1Toda}
P_1&=\int\left(u\theta\zeta_1-u\theta\zeta\right)\\ \label{eq:P2}
P_2&=\int\left(uu_1\theta\theta_1+uv_1\theta\zeta_1-uv\theta\zeta+u\zeta\zeta_1\right).
\end{align}
We look for a vector $X$ of the form $X=f(u,v)\theta+g(u,v)\zeta$ satisfying the equation $[P_1,X]=P_2-\alpha P_{(ul)}$. Note in particular that in our ansatz $f$ and $g$ depend only on the coordinates $(u,v)$ (i.e., not on their shifts).
An explicit computation with formula \eqref{eq:sch-def} gives
\begin{equation}
\begin{split}
[P_1,X]&=\int\left(-u_{-1}f_v\theta_{-1}\theta+(f-uf_u)\theta\zeta_1-(f-uf_u-ug_v)\theta\zeta\right.\\
&\qquad\left.-u_1g_v\theta_{-1}\zeta-ug_u\zeta\zeta_1\right)\\
&=\int\left(u(\cS f_v)\theta\theta_1+(f-uf_u-u\cS g_v)\theta\zeta_1-(f-uf_u-ug_v)\theta\zeta\right.\\&\qquad\left.-ug_u\zeta\zeta_1\right),
\end{split}
\end{equation}
where we denoted by $f_u$ the partial derivative of $f$ with respect to the variable $u$ (similar meaning for $f_v, g_u$ and $g_v$). By direct inspection we have the system of equations
\begin{equation}
\left\{\begin{array}{rl}
-u\cS f_v&=uu_1\\
f-uf_u-u\cS g_v &=uv_1\\
-f+uf_u+ug_v&=-uv -\alpha u\\
-ug_u&=u
\end{array}\right.
\end{equation}
which can be solved only if $\alpha=0$, obtaining $f=-uv$ and $g=-u-\frac12v^2$ and then confirming that $P_2$ is a coboundary. Recall that $[P_1,X]=\mathcal{L}_{-X}(P_1)$.
\begin{remark}
It is very easy to notice that, as $P_2$ is $d_{P_1}$-exact, also $P_1$ is $d_{P_2}$ exact: it is sufficient to note that the change of coordinates $u\mapsto u$, $v\mapsto v+\lambda$ produces $P_2\mapsto P_2+\lambda P_1$ (this corresponds to a local $1$- vector $X=-\lambda\int\zeta$). However, we cannot obtain from it any information on the cohomology of $P_2$.
\end{remark}
\subsection{Bruschi-Ragnisco lattice}\label{sec:BR}
The Bruschi-Ragnisco lattice equation is \cite{kmw13}
\begin{equation}
\left\{\begin{array}{rl} u_t&=u(u-u_{-1})\\ v_t&=uv_1-u_{-1}v\end{array}\right.
\end{equation}
and it is bi-Hamiltonian with respect to the first structure of the Toda lattice and to
\begin{equation}\label{eq:BR2}
H_2^{(BR)}=\begin{pmatrix}
0 & u(1-\cS^{-1})\circ u\\ u(\cS-1)\circ u & u\cS v-v\cS^{-1}\circ u
\end{pmatrix}
\end{equation}
which corresponds to the bivector
\begin{equation}\label{eq:BR2b}
P_2^{(BR)}=\int\left(u^2\theta\zeta+uu_1\zeta\theta_1+uv_1\zeta\zeta_1\right).
\end{equation}
The compatibility between $P_1$ and $P_2$ implies that $P_2^{(BR)}=[P_1,Y]+P_{(ul)}$. Note that now it is impossible to find a solution using the ansatz of Section \ref{sec:toda}, because for the coefficient of the term $\zeta\zeta_1$ we would have $g_u=-v_1$, but we have postulated that $g$ depends only on $(u,v)$. We then make the slightly more general ansatz that $f$ and $g$ depends on $(u_n,v_n)$ with $n=0,\pm1$. Computing the Schouten bracket of $P_1$ and $Y$ and comparing it with $P_2$ we obtain the following system of equations:
\begin{equation}
\left\{\begin{array}{rl}
u\cS^2 f_{v_{-1}}&=0\\
u\cS f_{v_{-1}}-u\cS f_v-u_1f_{v_1}&=0\\
u_1 f_{u_1}+u\cS^2 g_{v_{-1}}&=0\\
f-uf_u+u_1 f_{u_1}+u\cS g_{v_{-1}}-u\cS g_v&=0\\
-f+u f_u-u_{-1} f_{u_{-1}}+ug_v-ug_{v_1}&=u^2-\alpha u\\
-u\cS f_{u_{-1}}-u_1 g_{v_1}&=uu_1\\
-ug_u-u\cS g_{u_{-1}}+u_1g_{u_1}&=uv_1\\
u_1 g_{u_1}&=0
\end{array}\right.
\end{equation}
for which we find an easy solution for $\alpha=0$, namely $f=0$, $g=-uv_1$.
\subsection{Two-component Volterra lattice}\label{sec:5.3}
The two-component Volterra lattice \cite{kmw13}
\begin{equation}
\left\{\begin{array}{rl} u_t &=u(v_1-v)\\v_t&=v(u-u_{-1})\end{array}\right.
\end{equation}
is bi-Hamiltonian with respect to the two Hamiltonian operators
\begin{equation}\label{eq:2V-1}
H_1^{(2V)}=\begin{pmatrix} 0 & u(\cS-1)\circ v\\v(1-\cS^{-1})\circ u & 0\end{pmatrix}
\end{equation}
and
\begin{equation}\label{eq:2V-2}
H_2^{(2V)}=\begin{pmatrix} u(\cS v-v\cS^{-1})\circ u &  u(u\cS-u+\cS\circ  v-v)\circ v\\
v(u-\cS^{-1}\circ u+v-v\cS^{-1})\circ u & v(u\cS-\cS^{-1}\circ u)\circ v\end{pmatrix}
\,.
\end{equation}
The two associated bivectors are, respectively,
\begin{equation}\label{eq:P2V1}
P_1^{(2V)}=\int\left(uv_1\theta\zeta_1-uv\theta\zeta\right)
\end{equation}
and
\begin{equation}\label{eq:P2V2}
P_2^{(2V)}=\int\left(uu_1v_1\theta\theta_1+uv_1(u+v_1)\theta\zeta_1-uv(u+v)\theta\zeta+uvv_1\zeta\zeta_1\right)
\,.
\end{equation}
We observe that the change of coordinates $u'=\log u$, $v'=\log v$ brings $H_1^{(2V)}$ to the constant form $H_0$ in \eqref{eq:HToda-0}. Since the second cohomology of $P_0$ is concentrated in the ultralocal part, there exists a local 1-vector $X$ such that $P_2^{(2V)}=[P_1^{(2V)},X]+\alpha\, P'_{(ul)}$, where $P'_{(ul)}$ is obtained by the inverse change of coordinate for $H_{(ul)}$ as $P'_{(ul)}=\frac12\int uv\theta\zeta$. We look for the solution adopting the ansatz $X=f(u,v)\theta+g(u,v)\zeta$ and computing, as usual $P_2^{(2V)}=[P_1^{(2V)},X]$. We obtain the following set of equations
\begin{equation}
\left\{\begin{array}{rl}
-uv_1\cS f_v&=uu_1v_1\\ v_1 f+u\cS g-uv_1 f_u-uv_1\cS g_v&=uv_1(u+v_1)\\
-vf-ug+uvf_u+uvg_v&=-uv(u+v)-\alpha uv\\
-uv_1g_u&=uvv_1
\end{array}\right.
\end{equation}
From the first and fourth equation it is immediate to deduce $f=-uv + \tilde{f}(u)$ and $g=-uv+\tilde{g}(v)$. For degree reasons, a possible solution for $\tilde{f}$ and $\tilde{g}$ must be quadratic; it is then straightforward to find the solution $f=-uv-u^2$ and $g=-uv-v^2$, with the necessary and sufficient condition $\alpha=0$.
\begin{remark}
Note that the second Hamiltonian operator $H^{(2V)}_2$ is of the form $A\cS+B-\cS^{-1}\circ A^t$ with
$$
A=\begin{pmatrix}
u v_1 u_1 & u(u+v_1)v_1\\0 & uvv_1
\end{pmatrix}
$$
a nondegenerate matrix. According to Remark \ref{rem:DP}, then $H^{(2V)}_2$ must be defined by a triple $(\aff(\R),\rmat,\kmat)$
associated to the Poisson-Lie group $\aff(\R)$ and a choice of matrices $\rmat$ and $\kmat$. Under the change of coordinate $u'=u/v$, $v'=-1/v+u/v$, indeed, $H^{(2V)}_2$ takes the form
\begin{equation*}
\begin{split}
\widetilde{H}^{(2V)}_2&=\begin{pmatrix}0 & u'\\-u'u'_1 & -u'v'_1+u'\end{pmatrix}\cS+\begin{pmatrix}0 & -u'-(u')^2\\ u'+(u')^2& 0\end{pmatrix}\\&\quad-\cS^{-1}\circ\begin{pmatrix}0 &-u'u'_1\\ u' & -u'v'_1+u'\end{pmatrix}
\,,
\end{split}
\end{equation*}
namely $A$ and $B$ are respectively of the form \eqref{eq:Aaff} and \eqref{eq:Baff} for $\rmat=\begin{pmatrix}0&1\\-1&1\end{pmatrix}$ and $\kmat=\begin{pmatrix}0&1\\-1&0\end{pmatrix}$. Furthermore, the Lie bracket \eqref{eq:20240621-gstar} on $\mf g^*$ is $[e_1^*,e_2^*]=e_1^*$. 
\end{remark}

\subsection{Relativistic Volterra lattice}
The relativistic Volterra lattice \cite{kmw13}
\begin{equation}
\left\{\begin{array}{rl} u_t&=u(v_1-v+u_1v_1-uv) \\ v_t&=v(u-u_{-1}+uv-u_{-1}v_{-1})
\end{array}\right.
\end{equation}
is bi-Hamiltonian with respect to the same first Hamiltonian structure of the two-component Volterra equation $H_1^{(2V)}$ and to the $(-2,2)$ order structure
$$
H_2^{(rV)}=\begin{pmatrix}
a(\mc S)&b(\mc S)\\
-b^*(\mc S)&c(\mc S)
\end{pmatrix}
\,,
$$
where
\begin{align*}
&a(\mc S)= u\cS\circ u v(1+u)-uv(1+u)\cS^{-1}\circ u
\,,
\\
&b(\mc S)=
u(\cS\circ  uv\cS + u\cS+\cS\circ v)\circ v-uv(u+v+uv)
\,,
\\
&
c(\mc S)=uv(1+v)\cS\circ v-v\cS^{-1}\circ uv(1+v)
\,,
\end{align*}
which corresponds to the Poisson bivector
\begin{multline}
P_2^{(rV)}=\int\left(uu_1v_1(1+u_1)\theta\theta_1 + uu_1v_1v_2\theta\zeta_2+uv_1(u+v_1)\theta\zeta_1\right.\\
\left.-uv(u+v+uv)\theta\zeta+uv(1+v)v_1\zeta\zeta_1\right)
\,.
\end{multline}
The compatibility of $P_2^{(rV)}$ with $P_1^{(2V)}$ means that there exists a local $1$-vector $Y$ such that $P_2^{(rV)}-[P_1^{(2V)},Y]$ is ultralocal, in particular either zero or proportional to $P'_{(ul)}$ of the Section \ref{sec:5.3}. Moreover, we note that $P_2^{(rV)}=P_2^{(2V)}+\tilde{P}_2$, which -- by the linearity of the Schouten bracket -- implies that $Y=X+\tilde{X}$, with $X$ as computed in 
Section \ref{sec:5.3}. Since the order of $\tilde{P}_2$ is higher than the one of $P_1^{(2V)}$, we expect the local 1-vector $\tilde{X}$ to depend at least on the variables $(u,v,u_1,v_1)$; indeed, we can find a solution of such a form of the equation $\tilde{P}_2=[P_1^{(2V)},\tilde{X}]$. From the ansatz $\tilde{X}=f(u,v,u_1,v_1)\theta+g(u,v,u_1,v_1)\zeta$ we obtain the system of equations
\begin{equation}
\left\{\begin{array}{rl}
-u_1v_1f_{v_1}-uv_1\cS f_v&=uu_1^2v_1\\
-u_1v_2f_{u_1}&=uu_1v_1v_2\\
v_1 f+u\cS g-uv_1f_u+u_1v_1f_{u_1}-uv_1\cS g_v&=0\\
-vf-ug+uvf_u+uvg_v-uv_1g_{v_1}&=-u^2v^2-\alpha u^2\\
-uv_1g_u+u_1v_1g_{u_1}&=uv^2v_1\\
-u_1v_1g_{u_1}&=0
\end{array}\right.
\end{equation}
which has solution $f=-uu_1v_1$ and $g=-uv^2$ only if $\alpha=0$. Then we reconstruct the local $1$-vector $Y$ as
\begin{equation}
Y=-\left(u^2+uv+uu_1v_1\right)\theta-\left(v^2+uv+uv^2\right)\zeta
\end{equation}
and conclude $P_2^{(rV)}=[P_1^{(2V)},Y]$.
\subsection{Relativistic Toda lattice} 
The relativistic Toda lattice \cite{kmw13}
\begin{equation}
\left\{\begin{array}{rl}
u_t&=u\left(u_{-1}-u_{1}+v-v_1\right)\\
v_t&=v\left(u_{-1}-u\right)
\end{array}\right.
\end{equation}
is bi-Hamiltonian with respect to the Hamiltonian operators
\begin{equation}\label{eq:HrT1}
H^{(rT)}_1=\begin{pmatrix} 0 & u(\cS-1)\\ (1-\cS^{-1})\circ u & \cS^{-1}\circ u-u\cS\end{pmatrix}
\end{equation}
and
\begin{equation}\label{eq:HrT2}
H^{(rT)}_2=\begin{pmatrix} u(\cS-\cS^{-1})\circ u & u(\cS-1)\circ v\\v(1-\cS^{-1})\circ u & 0\end{pmatrix}
\,.
\end{equation}
\begin{remark}
Note that the pair $(H_1^{(rT)}, H_2^{(rT)})$ can be obtained by the pair $(H,H_2)$ of the Toda lattice (\eqref{eq:HToda} and \eqref{eq:HToda2}) by noting $H_1^{(rT)}=H-\tilde{H}$ and $H_2^{(rT)}=H_2-\tilde{H}$ with
\begin{equation}
\tilde{H}=\begin{pmatrix}0 & 0 \\ 0 & u\cS-\cS^{-1}\circ u\end{pmatrix}
\,.
\end{equation}
In fact, $(H, H_2, \tilde{H})$ is a triple of compatible Hamiltonian operators. 
\end{remark}
The first Hamiltonian operator in \eqref{eq:HrT1}, whose associated bivector is
\begin{equation}\label{eq:PrT1}
P_1^{(rT)}=\int\left(u\theta\zeta_1-u\theta\zeta-u\zeta\zeta_1\right)
\end{equation}
can be brought to the constant form \eqref{eq:HToda-0} by the change of coordinates $u'=\log u$, $v'=u+v$; this means that there must exist a local 1-vector $X$ such that, when written in the same coordinate system as $P_1^{(rT)}$, is a solution of $[P_1^{(rT)},X]=P_2^{(rT)}-\alpha P_{(ul)}$, where $P_2^{(rT)}$ is the Poisson bivector associated to \eqref{eq:HrT2}, namely
\begin{equation}\label{eq:PrT2}
P^{(rT)}_2=\int\left(uu_1\theta\theta_1+uv_1\theta\zeta_1-uv\theta\zeta\right),
\end{equation}
 and $P_{(ul)}$ is the same ultralocal bivector of Section~\ref{sec:toda} and Section~\ref{sec:BR}, as one can easily verify by applying the inverse change of coordinates to $\frac12\int\theta\zeta\in H^2(\hF,d_{P_0})$.
To find a solution we look for a more generic ansatz than the previous examples (there are not solutions choosing simpler ones) of the form $X=f\theta+g\zeta$ with both $f$ and $g$ depending on $(u_n,v_n)$, $n=0,\pm 1$. The system of conditions we obtain is
\begin{equation}
\left\{\begin{array}{rl}
u \cS^{2}f_{v_{-1}}&=0\\
u_1 f_{v_1}+u\cS f_v&=uu_1\\
-u_1 f_{v_1}+u_1 f_{u_1}+u\cS^2 g_{v_{-1}}&=0\\
-f-u f_v-u_1 f_{u_1}+uf_u+u\cS g_v-u\cS g_{v_{-1}}&=uv_1\\
f+u f_v -u f_u+u_{-1}f_{u_{-1}}+ug_{v_1}-ug_v&=-uv-\alpha u\\
u_{-1}f_v-u_{-1}f_{u_{-1}}-u \cS^{-1} g_{v_1}&=0\\
-u_1 g_{v_1}-u\cS^2 g_{v_{-1}}+u_1 g_{u_1}&=0\\
f-ug_v-u\cS g_v-u_1 g_{u_1}+ug_u+ug_{u_{-1}}&=0
\end{array}\right.
\end{equation}
whose solution is $f=-uv-uu_1$, $g=uv_1-u_{-1}v_{-1}+uv+\frac12(u^2+v^2)$  with $\alpha=0$.

\subsection{Derivation of bi-Hamiltonian pairs}
Our result about the structure of $H^2(\hF,d_{P_0})$ guarantees that any bivector $P$ \emph{compatible} with $P_0$ is of the form $P=\alpha P_{(ul)}+[P_0,X]$. Therefore, given an arbitrary local 1-vector $X$ and a constant $\alpha$, we can obtain a new compatible bivector; however, there is no guarantee that the bivector $P$ is Poisson and therefore that the pair $(P_0,P)$ is bi-Hamiltonian. There exist choices of $X$ for which this is indeed the case: the pairs of the previous section are some examples.

In this section, we characterize the local 1-vectors $X$ that produce compatible Poisson bivectors and present a simple class of new examples. For simplicity, we work with $P_0$; however, any point transformation (such as those connecting $P_0$ with $P_1$ of the previous examples) applied to both the bivectors preserve the bi-Hamiltonian property.

\begin{proposition}
Let $P_0$ be given by \eqref{eq:PToda}, let $X\in\hF^1$ and let $P=\alpha P_{(ul)}+[P_0,X]$. The pair $(P_0,P)$  is a bi-Hamiltonian pair if and only if \begin{equation}\label{eq:biHam-thm}
\left[X,\left[X,P_0\right]-2\alpha P_{(ul)}\right]=\beta P_{(ul)}+[P_0,Y],\end{equation}
for some value of the constant $\beta$ and a local $1$-vector $Y$.
\end{proposition}
\begin{proof}
We have already proved that $P_0$ is a Poisson bivector and that $P_0$ and $P$ are compatible, so we only need to impose the Poisson property for $P$, which reads
$$
[\alpha P_{(ul)}+[P_0,X],\alpha P_{(ul)}+[P_0,X]]=2\alpha[P_{(ul)},[P_0,X]]+[[P_0,X],[P_0,X]]=0.
$$
Using the graded Jacobi identity of the Schouten bracket, we rewrite
\begin{align*}
[P_{(ul)},[P_0,X]]&=[[P_{(ul)},P_0],X]-[P_0,[P_{(ul)},X]]=-[P_0,[P_{(ul)},X]],\\
[[P_0,X],[P_0,X]]&=[[[P_0,X],P_0],X]-[P_0,[[P_0,X],X]]=-[P_0,[[P_0,X],X]],
\end{align*}
and, by graded skewsymmetry,
\begin{equation}\label{eq:biHam-proof}
[P,P]=[P_0,2\alpha [X,P_{(ul)}]-[X,[X,P_0]]]=0.
\end{equation}
Condition \eqref{eq:biHam-proof} is satisfied if and only if the second entry of the Schouten bracket is a cocycle of $d_{P_0}$ which, by our Theorem~\ref{thm:cohoToda}, imposes the form $\beta P_{(ul)}+[P_0,Y]$.
\end{proof}
Local 1-vectors $X$ satisfying this property are, for example, those whose characteristics are linear in the $(u,v)$ variables and their shifts, namely of the form $X^i=A^{im}_j u^j_m$. It is quite cumbersome to observe condition~\eqref{eq:biHam-thm} directly; however, the bivector $[P_0,X]$ does not depend on $u,v$ and their shifts, and therefore it is always Poisson and compatible with $P_{(ul)}$, so that $\alpha P_{(ul)}+[P_0,X]$ is Poisson, too. Explicitly we have 
$$
P=\alpha P_{(ul)}+[P,X]=\begin{pmatrix} a(\cS)&b(\cS)+\alpha \\-b^*(\cS)-\alpha&c(\cS)\end{pmatrix}
$$
where
\begin{align*}
a(\cS)&=A^{1m}_2\left(\cS^{m-1}-\cS^m+\cS^{-m}-\cS^{1-m}\right),\\
b(\cS)&=A^{1m}_1\left(\cS^m-\cS^{m+1}\right)+A^{2m}_2\left(\cS^{1-m}-\cS^{-m}\right),\\
c(\cS)&=A^{2m}_1\left(\cS^m-\cS^{m+1}+\cS^{-m-1}-\cS^{-m}\right).
\end{align*}
 A very simple example to illustrate  this family can be obtained for any $\alpha$ and $X$ such that $A^{10}_2$ is the only non-vanishing parameter; in this case $X$ satisfies the stronger properties $[X,[X,P_0]]=0$ and $[X,P_{(ul)}]=0,$ from which \eqref{eq:biHam-thm} is solved by $\beta=Y=0$. The bi-Hamiltonian pairs are of the form
$$
P_0=\begin{pmatrix} 0 &\cS-1\\1-\cS^{-1}&0\end{pmatrix},\qquad P=\begin{pmatrix}-A^{10}_2\left(\cS-\cS^{-1}\right)&\alpha\\-\alpha & 0\end{pmatrix}
$$
which, in the same local coordinates as the first Hamiltonian structure of Toda, are
\begin{equation}
P_1=\begin{pmatrix}0 & u(\cS-1)\\(1-\cS^{-1})\circ u& 0\end{pmatrix},\qquad P_2=\begin{pmatrix}-A^{10}_2\, u(\cS-\cS^{-1})\circ u & \alpha\, u\\-\alpha\,u&0\end{pmatrix}.
\end{equation}

Other classes of examples can be obtained imposing different ansatz for the local 1-vector $X$, and selecting among all the possible compatible bivectors $[P_0,X]$ those which satisfy the Jacobi identity. If we choose $X=X(u,v)$ at most cubic in its variables, we obtain two non-constant  families of Hamiltonian operators compatible with $P_0$, namely
$$
P_{(1)}=\begin{pmatrix} 0 & p_2(u)(\cS-1)\\(1-\cS^{-1})\circ p_2(u)& q_2(u)\cS-\cS^{-1}\circ q_2(u) \end{pmatrix}
$$
and
$$
P_{(2)}=\begin{pmatrix} \cS\circ p_2(v)-p_2(v)\cS^{-1} & q_2(v)(1-\cS^{-1})\\(\cS^{-1}-1)\circ q_2(v)& 0 \end{pmatrix}
$$
with $p_2(x)$, $q_2(x)$ arbitrary quadratic polynomials of one variable.  Observe that, despite the similar structure, $P_{(1)}$ and $P_{(2)}$ are not interchangeable keeping $P_0$ constant. Indeed, the transformations preserving $P_0$ are only those of the form $\tilde{u}=a\,u+b$, $\tilde{v}=a^{-1}v+c$, and those can only change the coefficients of the polynomials $p_2$, $q_2$, not their position or variable dependency.


\section{Final remarks}

In this work we addressed the study of multi-component difference Hamiltonian operators, extending to the $\ell=2$ case the investigation on their structure carried out in \cite{DSKVW18-1,DSKVW18-2} and on their Poisson cohomology presented in \cite{CW19}. About the former problem, we have insofar limited ourselves to the ``first neighbours'' operators of order $(-1,1)$. While Dubrovin's result \cite{Dub89} holds true for any number of components, it is worthy noticing that most of the differential-difference integrable systems studied in the literature have a leading order defined by degenerate matrices and therefore fall outside the scope of Dubrovin's work. In Proposition \ref{prop:deg-B} and Proposition \ref{prop:deg-noB} we listed the $\ell=2$ degenerate cases depending on once-shifted variables.

Our computation of the Poisson cohomology for $\ell=2$ Hamiltonian operators essentially carries the same message as the previously studied scalar case \cite{CW19}: the only nontrivial part of the Poisson cohomology is concentrated in the so-called ultralocal part, and hence do not exist ``dispersive'' (such a term is borrowed from the analogue differential case: strictly speaking, of order $(-N,N)$ for $N>1$) deformations of $(-1,1)$-order Hamiltonian operators. This hints to a prominent role, within the theory of differential-difference integrable systems, for operators of $(-1,1)$-order, as it is indeed demonstrated by the abundance of examples in such a class \cite{kmw13}.

As it was already observed in \cite{DSKVW18-1}, the $n$-stretched versions, namely the ones obtained from the substitution $(\cS^k\mapsto \cS^{nk}, u^i_k\mapsto u^i_{nk})$, of all the Hamiltonian operators in our classification are Hamiltonian operators, too. However, we expect that their corresponding Poisson cohomology becomes richer, in particular allowing nontrivial ``dispersive'' deformations. Moreover, as the order of the Hamiltonian operators grows, also the classification will grow richer, as it is already the case for the $\ell=1$ case -- indeed, there exist examples \cite{kmw13} of integrable differential--difference systems in two components with bi-Hamiltonian structure of order $(-2,2)$ that, for the moment, fall outside the scope of our classification.

Finally, the class of Hamiltonian operators we considered in this paper is limited to the \emph{local} ones. Nonlocal Hamiltonian operators in the $\ell=1$ case have been considered in \cite{DSKVW18-2,cmw20}. An extension of their results to the $\ell>1$ case is of interest, not only to better investigate multi-component differential-difference integrable systems as the ones considered in this work, but also in the direction of non-commutative systems, whose Hamiltonian structure, even in the simplest cases, is often nonlocal \cite{CW21, CW22, FV23}.


\medskip

{\bf Acknowledgements.}
M.~C.~is a member of the GNFM INdAM group and he is supported by the National Science Foundation of China (Grant no.~12101341 and no.~12431008), Ningbo City Yongjiang Innovative Talent Program and Ningbo University Talent Introduction and Resarch Initiation Fund. He is also grateful to the Department of Mathematics ``G.~Castelnuovo'' of Sapienza University of Rome for their kind hospitality during several visits.
D.~V. is a member of the GNSAGA INdAM group, he has been supported by the National PRIN 2022 Grant 2022HMBTTL and he acknowledges the financial support of INFN, IS CSN4 MMNLP.


\end{document}